\documentclass[showkeys]{revtex4}

 \usepackage{lineno,hyperref}
 
\usepackage{amsmath}             
\usepackage{amssymb}
\usepackage{mathtools}
\usepackage{amsthm}             
\theoremstyle{plain}             
\newtheorem{theorem}{Theorem}[section]

\theoremstyle{definition}

\newcommand{\ma}{\textit{Mathematica$^{\small{\circledR}}$}}

\begin{document}

\title{Linearization and Krein-like functionals of hypergeometric orthogonal polynomials}

\author{J.S. Dehesa}
\email[]{dehesa@ugr.es}
\affiliation{Departamento de F\'isica At\'omica, Molecular y Nuclear, Universidad de Granada, Granada 18071, Spain}
\affiliation{Instituto Carlos I de F\'isica Te\'orica y Computacional, Universidad de Granada, Granada 18071, Spain}

\author{J. J. Moreno-Balc\'{a}zar}
\affiliation{Departamento de  Matem\'{a}ticas,  Universidad de Almer\'{\i}a,  Almer\'{\i}a 04120, Spain}
\affiliation{Instituto Carlos I de F\'isica Te\'orica y Computacional, Universidad de Granada, Granada 18071, Spain}

\author{I.V. Toranzo}
\affiliation{Departamento de Matem\'aticas, Universidad Rey Juan Carlos, Madrid 28933, Spain}

\begin{abstract}
The Krein-like $r$-functionals of the hypergeometric orthogonal polynomials $\{p_{n}(x) \}$ with kernel of the form $x^{s}[\omega(x)]^{\beta}p_{m_{1}}(x)\ldots p_{m_{r}}(x)$, being $\omega(x)$ the weight function on the interval $\Delta\in\mathbb{R}$, are determined by means of the Srivastava linearization method. The particular $2$-functionals, which are particularly relevant in quantum physics, are explicitly given in terms of the degrees and the characteristic parameters of the polynomials. They
include the well-known power moments and the novel Krein-like moments. Moreover, various related types of exponential and logarithmic functionals are also investigated.
\end{abstract}

\keywords{
Krein functionals, hypergeometric functions, orthogonal polynomials, linearization method, Lauricella functions
}

\maketitle

\section{Introduction}

Since the works of Ferrers and Adams in the late $1870$s and Bailey, Dougall and Erd\'elyi in the early $1900$s about the linearization formula of the product of two hypergeometric orthogonal polynomials (HOPs) of Legendre \cite{ferrers1877,adams1878,bailey1933}, Gegenbauer or ultraspherical \cite{dougall1919} and Laguerre \cite{erdelyi1938} type up until now, an intense activity about the linearization problem has been developed in the theory of HOPs from both theoretical \cite{askey1971,rahman1981,niukkanen1985,srivastava1988,ruiz1997,artes1998, ruiz1999,hounkonnou2000,leo2001,ruiz2001,ruiz2001b,srivastava2003a,srivastava2003b,park2006,sanchez2013} and applied \cite{koepf1998,koepf1998a,area2001,koepf2002,chen2003,gautschi2004,buyarov2004,gil2007,chaggara2007,koepf2010,foupou2013,tcheutia2014,evnin2016,savin2017,gautschi2018} standpoints. See also the monographs of Askey \cite{askey1975}, Andrews-Askey-Roy \cite{andrews1999}, Ismail \cite{ismail2005}, Koekoek-Lesky-Swarttouw \cite{koekoek2010}, Gautschi \cite{gautschi2004,gautschi2018} and Tcheutia \cite{tcheutia2014} for partial periodic reviews.\\

This linearization problem, which is equivalent to the evaluation of integrals of the product of three or more HOPs of the same type, is also called Clebsch-Gordan-type problem because its structure is similar to the Clebsch-Gordan series for spherical functions \cite{edmonds1957,niukkanen1985,srivastava1988,alvarez1998}. It arises in a wide range of mathematical and physical problems, from the generalized moment problems as stated by M. G. Krein \cite{akhiezer1962,krein1977,krein1959}, stochastic processes \cite{anshelevich2000} and combinatorics \cite{even1976,azor1982,sainte1985,ismail2013}, information entropies \cite{sanchez2013,larsson2002,sanchez2010,sanchez2011,puertas2017,puertas2017b} up to quantum physics. Indeed, numerous physical and chemical properties of a given complex system (atoms,  molecules,...) are determined by these overlap integrals or Krein-like functionals of various HOPs, basically because the main route for the search of the solutions of the corresponding Schr\"odinger equation is their expansion as linear combinations of a known basis set of functions which are controlled by the HOPs (see e.g. \cite{niukkanen1985, avery1992,aquilanti2001,avery2004,mitnik2011,avery2012,coletti2013,mccoy2016}). The computation of these overlap Krein-like integrals is a formidable task both analytically and numerically; the latter is basically because a naive numerical evaluation using quadratures is not convenient due to the increasing number of integrable singularities when the polynomial degree is increasing, which spoils any attempt to achieve reasonable accuracy even for rather small degree (see e.g., \cite{buyarov2004}).\\

In this work we first determine the generalized Krein-like $r$-integral functionals of HOPs of the form
\begin{equation}
\label{RFR}
\mathcal{J}_{\{m_{r} \}}(s,\beta) := \int_{\Delta}  [\omega(x)]^{\beta}\,x^{s}\, p_{m_{1}}(x)\,\ldots \,p_{m_{r}}(x)\, dx,
\end{equation}
where $s$ and $\beta$ are real parameters, and  $\omega(x)$ denotes the weight function on the real interval $\Delta$ with respect to which the polynomials $\{p_{m}(x) \}$ are orthogonal. Note that when 
$r=2$, they simplify to the standard Krein-like functionals mentioned above. These functionals $\mathcal{J}_{\{m_{r} \}}(s,\beta)$ are encountered in the computation of numerous fundamental and/or experimentally accessible quantities (e.g., multilevel transition probabilities, electric multipole moments, and kinetic, exchange and correlation energies) of atomic, molecular and nuclear systems by means of entropic and density functional methods (see e.g., \cite{barret1979,hasse1988,parr1989,ghosh2000} et sequel). The resulting compact expressions are given in the form of a multivariate Lauricella's hypergeometric function \cite{lauricella1893,appell1926,srivastava1985,srivastava1988} $F_{A}^{(r)}(x_1, ...,x_r)$ evaluated at $(\frac{1}{\beta}, ...,\frac{1}{\beta})$. Then, we apply these expressions to the Laguerre, Hermite and Jacobi polynomials \cite{nikiforov1988,olver2010} . Moreover, we use the corresponding results to evaluate various relevant mathematical quantities which include the well-known power moments $\langle x^{s}\rangle_{n}$, the Krein-like moments $\langle [\omega(x)]^{k}\rangle_{n}$, and various related types of exponential and logarithmic functionals of the form $\langle x^{k}e^{-\alpha x}\rangle_{n}$, $\langle (\log x)^{k}\rangle_{n}$ and $\langle [\omega(x)]^{k}\log \omega(x)\rangle_{n}$, where $\langle f(x)\rangle_{n} :=\int_{\Delta} f(x)\rho_{n}(x)\, dx$ and $\rho_{n}(x) = \omega(x)\,[p_{n}^{(\alpha)}(x)]^2$ is the Rakhmanov probability density \cite{rakhmanov1977} associated to the polynomial $p_{n}^{(\alpha)}(x)$. This density plays a relevant role both in approximation theory and quantum physics because it controls the asymptotic behavior of the ratio of polynomials with consecutive degrees \cite{rakhmanov1977} and it describes the quantum probability density of numerous single-particle quantum systrems (see e.g., \cite{dehesa2001}). The above-mentioned mathematical expectation values  describe a number of entropic \cite{ruiz2000} and physical quantities which are experimentally accessible in atomic and nuclear physics \cite{barret1970,barret1979,hasse1988,ford1973,engfer1974,duch1983,surzhykov2005,suslov2008}. \\

Then, we calculate the simplest and most familiar Krein-like $2$-functionals of the form
\begin{equation}
\label{RF2}
\mathcal{J}_{m,n}(s,\beta):=\int_{0}^{\infty} [\omega(x)]^{\beta}\,x^{s}\,p_{m}(x)\,p_{n}(x)\, dx
\end{equation}
by means of two different approaches. One which gives the functionals in terms of the coefficients of the second-order differential equation fulfilled by the HOP under consideration. Another one which allows us to express the functionals explicitly in terms of the degrees and the characteristic parameters of the polynomials by making use of various algebraic  characterizations of the HOPs, if known; this is illustrated for the canonical families of Hermite, Laguerre and Jacobi types. 
Both methods are applied to various mathematical quantities of the above-mentioned form $\langle f(x)\rangle_{n}$.

The structure of this work is the following. First, we compute the general functionals $\mathcal{J}_{\{m_{r} \}}(s,\beta)$ given by Eq. \eqref{RFR} by using the scarcely known Srivastava-Niukkanen's linearization method for the product of a finite number of HOPs \cite{niukkanen1985,srivastava1988}, where in most cases the linearization coefficients are given in terms of the Lauricella function of type A \cite{lauricella1893,appell1926}. Then, this method is used for the Laguerre, Hermite and Jacobi polynomials and later the resulting compact expressions are applied for the evaluation of various mathematical moments of power, Krein, exponential and logarithmic types. Second, we compute the functionals $\mathcal{J}_{m,n}(s,\beta)$, given by Eq. \eqref{RF2}, by using the second-order hypergeometric differential equation of the involved HOPs, obtaining expressions in terms of the polynomial coefficients of the differential equation. Then, the utility of these general expressions is illustrated by applying them to the Laguerre, Hermite and Jacobi polynomials. Third, we compute the functionals $\mathcal{J}_{m,n}(s,\beta)$, given by Eq. \eqref{RF2}, for the Laguerre, Hermite and Jacobi polynomials by means of various characterizations of these HOPs, obtaining explicit expressions in terms of the degrees and the characteristic parameters of the associated weight functions; then, these expressions are applied for the evaluation of the mathematical moments mentioned above.  \\

\section{Lauricella-based approach to generalized Krein-like functionals of HOPs}

In this section, we describe a method to compute the generalized Krein-like integral  functionals of HOPs $\{p_{m}\}\equiv\{p_{m}^{(\alpha)}\}$, $\mathcal{J}_{\{m_{r} \}}(s,\alpha,\beta)$, defined in Eq. \eqref{RFR}, by using the orthogonality condition of the polynomials and the following Srivastava-Niukkanen's expansion \cite{niukkanen1985,srivastava1988} of the integral kernel
\begin{equation}
\label{SBLM1}
x^{s}p_{m_{1}}^{(\alpha)}(x)\ldots p_{m_{r}}^{(\alpha)}(x) = \sum_{i=0}^{\infty} c_{i}(s,r,\{m_{1}, \ldots, m_{r} \}, \alpha) p_{i}^{(\gamma)}(x),
\end{equation}
where the linearization coefficients $c_{i}(s,r,\{m_{1}, \ldots, m_{r} \},\alpha)$ can be usually expressed in terms of a multivariate hypergeometric function of Lauricella \cite{lauricella1893,appell1926} or Srivastava-Daoust types \cite{srivastava1988}.   \\
We take \eqref{SBLM1} into \eqref{RFR} obtaining
\begin{equation}
\label{GKLF2}
\mathcal{J}_{\{m_{r} \}}(s,\alpha,\beta)  = \sum_{i=0}^{\infty} c_{i}(s,r,\{m_{1}, \ldots, m_{r} \}, \alpha)  \int_{\Delta} [\omega(x)]^{\beta}\, p_{i}^{(\gamma)}(x)\, dx,
\end{equation}
which by a proper change of variable can be transformed as
\begin{equation}
\label{GKLF3}
\mathcal{J}_{\{m_{r} \}}(s,\alpha,\beta)  = \sum_{i=0}^{\infty} \tilde{c}_{i}(s,r,\{m_{1}, \ldots, m_{r} \}, \alpha,\beta)  \int_{\Delta} \omega(y)\, p_{i}^{(\gamma)}(y)\, dy
\end{equation}
Then, we use the orthogonality relation of the involved HOPs, i.e.,
\begin{equation}
\label{OR}
\int_{\Delta} \omega(x)p^{(\alpha)}_{n}(x)p^{(\alpha)}_{m}(x)\, dx = d_{n}^{2}\,\delta_{n,m},
\end{equation}
where $d_{n}$ stands for the normalization constant, to finally obtain
\begin{equation}
\label{GKLF4}
\mathcal{J}_{\{m_{r} \}}(s,\alpha,\beta)  =  \tilde{c}_{0}(s,r,\{m_{1}, \ldots, m_{r} \}, \alpha,\beta)\,d_{n}^{2}.
\end{equation}
Herein, we have assumed $p_{0}^{(\gamma)}(x)=1$ and applied Eq. \eqref{OR}. The explicit form of the coefficients $ \tilde{c}_{i}(s,r,\{m_{1}, \ldots, m_{r} \}, \alpha,\beta)$ depends on the previous change of variable.\\

Let us now use this method for generalized Krein-like integral functionals of the three canonical families of HOPS in a continuous real variable; namely, the Laguerre, Hermite and Jacobi polynomials.

\subsection{Laguerre polynomials}

In this case the method gives the following result:
\begin{theorem}
Let $L_{m}^{(\alpha)}(x)$ denote the Laguerre polynomials orthogonal with respect to the weight function $\omega_{\alpha}(x)=x^{\alpha}e^{-x}$ on $(0,\infty)$ \cite{nikiforov1988,olver2010}. Then, the generalized Krein-like functionals of the Laguerre polynomials defined by 
\begin{equation}
\label{GKLFL1}
\mathcal{J}_{\{m_{r} \}}^{(L)}(s,\alpha,\beta)  = \int_{0}^{\infty} [\omega_{\alpha}(x)]^{\beta} x^{s}L_{m_{1}}^{(\alpha)}(x)\cdots L_{m_{r}}^{(\alpha)}(x)\, dx, \quad \alpha >-\frac{s+1}{\beta},\quad s\in\mathbb{R}_{+}
\end{equation}
are given by
\begin{align}
\label{GKL}
\mathcal{J}^{(L)}_{ \{m_{r} \} } (s,\alpha,\beta) &= \beta^{-\beta\alpha-s-1}\,c_{0}\left( \beta\alpha+s, r, \{m_{r} \}, \frac{1}{\beta}, \alpha,0 \right). 
\end{align}
with
\begin{align}
\label{SNLC2}
 c_{0}(\beta\alpha+s,r,\{m_{r} \}, 1/\beta, \alpha,0) &= \Gamma(\beta\alpha +s +1) \binom{m_{1}+\alpha}{m_{1}}\cdots \binom{m_{r}+\alpha}{m_{r}}\nonumber\\
& \times F_{A}^{(r)} \left(\begin{array}{cc}
														 \beta\alpha +s  +1 ; -m_{1}, \ldots, -m_{r}& \\
																										&; \frac{1}{\beta}, \ldots, \frac{1}{\beta}\\
														\alpha+1, \ldots, \alpha+1 & \\
														\end{array}\right).
\end{align}

\end{theorem}
\begin{proof}
To obtain \eqref{GKL} from \eqref{GKLFL1} we first make the change of variable $y=\beta x$, obtaining
\begin{align}
\label{GKLFL2}
\mathcal{J}_{\{m_{r} \}}^{(L)}(s,\alpha,\beta)  &= \int_{0}^{\infty} e^{-\beta x}x^{\beta\alpha+s}L_{m_{1}}^{(\alpha)}(x)\ldots L_{m_{r}}^{(\alpha)}(x)\, dx \nonumber \\
&= \beta^{-\beta\alpha-s-1}\int_{0}^{\infty} e^{-y}y^{\beta\alpha+s}L_{m_{1}}^{(\alpha)}(y/\beta)\ldots L_{m_{r}}^{(\alpha)}(y/\beta)\, dy
\end{align}
Now we use the following Srivastava-Niukkanen expansion \cite{niukkanen1985,srivastava1988} for the Laguerre product  
\begin{equation}
\label{SNL}
y^{\mu}L_{m_{1}}^{(\alpha)}(ty)\cdots L_{m_{r}}^{(\alpha)}(ty) = \sum_{i=0}^{\infty} c_{i}(\mu, r,t,\{m_{r} \},\alpha,\gamma)\,L_{i}^{(\gamma)}(y)
\end{equation} 
($\gamma>-1$) with the linearization coefficients
\begin{align}
\label{SNLC1}
c_{i}(\mu, r,t,\{m_{r} \},\alpha,\gamma) &= (\gamma +1)_{\mu} \binom{m_{1}+\alpha}{m_{1}}\cdots \binom{m_{r}+\alpha}{m_{r}} \nonumber \\
&\times F_{A}^{(r+1)} \left(\begin{array}{cc}
														\gamma + \mu +1 ; -m_{1}, \ldots, -m_{r}, -i &\\
																										&; t, \ldots, t, 1 \\
														\alpha+1, \ldots, \alpha+1, \gamma +1 & \\
														\end{array}\right)
\end{align}
where $F_{A}^{(r+1)}(x_{1},\ldots,x_{r})$ denotes the following Lauricella function of type A of $r+1$ variables and $2r+3$ parameters \cite{srivastava1985}
\begin{equation}
\label{LF1}
F_{A}^{(s)} \left(\begin{array}{cc}
														a ; b_{1}, \ldots, b_{s} &\\
																										&; x_{1}, \ldots, x_{s} \\
														c_{1}, \ldots, c_{s} & \\
														\end{array}\right) = \sum_{j_{1},\ldots,j{s}=0}^{\infty}\frac{(a)_{j_{1}+\ldots+j_{s} }(b_{1})_{j_{1} } \cdots (b_{s})_{j_{s} } }{(c_{1})_{j_{1} } \cdots (c_{s})_{j_{s} } }\frac{x_{1}^{j_{1}}\cdots x_{s}^{j_{s}} }{j_{1}!\cdots j_{s}!}.
\end{equation}

  Then, we use \eqref{SNL} with $\mu=\beta\alpha+s$, $t=1/\beta$ and $\gamma=0$ into \eqref{GKLFL2} to obtain
\begin{equation}
\label{GKLFL3}
\mathcal{J}_{\{m_{r} \}}^{(L)}(s,\alpha,\beta)  = \beta^{-\beta\alpha-s-1}\sum_{i=0}^{\infty} c_{i}(\beta\alpha+s,r,\{m_{r} \}, 1/\beta, \alpha,0) \int_{0}^{\infty} e^{-y}L_{i}^{(0)}(y)\, dy.
\end{equation} 
Finally, the orthogonality condition of the Laguerre polynomials \cite{olver2010} allows us to find the wanted Eq. \eqref{GKL}  since the only term which survives in the infinite summation is the one with $i=0$.
\end{proof}
\textbf{Applications.}
Let us now apply this theorem to the power, Krein-like, exponential and logarithmic moments of the Rakhmanov probability density $\rho_{n}^{(L)}(x)$ of the Laguerre polynomials $L_{n}^{(\alpha)}(x)$ given by
$$\rho_{n}^{(L)}(x) = \omega_{\alpha}(x)\,[L_{m}^{(\alpha)}(x)]^2$$
\begin{enumerate}
	\item Power moments and Krein-like moments.
	\\
The power $\langle x^{s}\rangle_{n}$ and Krein-like $\langle [\omega(x)]^{k}\rangle_{n}$ moments of the Laguerre polynomials are given by
\begin{equation}
\label{LPM1}
\langle x^{s}\rangle_{n}^{(L)} = \int_{\Delta} x^{s}\rho_{n}^{(L)}(x)\, dx = \mathcal{J}_{n,n}^{(L)}(s,\alpha,1),
\end{equation}
\begin{equation}
\label{LKM1}
\langle [\omega(x)]^{k}\rangle_{n}^{(L)} = \int_{\Delta}  [\omega(x)]^{k}\rho_{n}^{(L)}(x)\, dx = \mathcal{J}_{n,n}^{(L)}(0,\alpha,k+1)
\end{equation}
respectively. From Eq. \eqref{GKL} one has that 
 \begin{equation}
 \label{LPM2}
 \langle x^{s}\rangle_{n}^{(L)} = c_{0}\left( \alpha+s, 2,n, 1, \alpha,0 \right),
 \end{equation}
 with 
 \begin{equation}
 \label{PMCS1}
 c_{0}\left( \alpha+s, 2,n, 1, \alpha,0 \right) = \Gamma(\alpha +s +1) \binom{n+\alpha}{n}^{2}F_{A}^{(2)} \left(\begin{array}{cc}
 														 \alpha +s  +1 ; -n, -n& \\
 																										&; 1,1 \\
 														\alpha+1, \alpha+1 & \\
 														\end{array}\right),
 \end{equation}
 and
 \begin{equation}
 \label{LKM2}
 \langle [\omega(x)]^{k}\rangle_{n}^{(L)} = (k+1)^{-(k+1)\alpha-1}c_{0}\left( (k+1)\alpha, 2, n, \frac{1}{k+1}, \alpha,0 \right)
 \end{equation}
 with
  \begin{align}
  \label{PMCS2}
  c_{0}\left( (k+1)\alpha, 2, n, \frac{1}{k+1}, \alpha,0 \right) = \Gamma((k+1)\alpha  +1) \binom{n+\alpha}{n}^{2}\nonumber \\
  & \hspace{-2cm} \times F_{A}^{(2)} \left(\begin{array}{cc}
  														 (k+1)\alpha +1 ; -n, -n& \\
  																										&; \frac{1}{k+1}, \frac{1}{k+1} \\
  														\alpha+1, \alpha+1 & \\
  														\end{array}\right),
  \end{align}
 for the power and the Krein-like moments, respectively, where $F_{A}^{(2)} \equiv F_{2}$ stands for the Appell hypergeometric series of two variables.
 
 \item Other exponential and logarithmic functionals.
 \\
 The logarithmic and exponential functionals of the Laguerre polynomials are defined as
\begin{equation}
\label{LLM1}
\langle (\log x)^{k}\rangle_{n}^{(L)} := \int_{\Delta}  (\log x)^{k}\rho^{(L)}_{n}(x)\, dx = \frac{d^{k}}{ds^{k}}\mathcal{J}_{n,n}^{(L)}(s,\alpha,1)\Bigg|_{s=0},
\end{equation}
\begin{equation}
\label{EFML1}
\langle x^{k}e^{-\alpha x}\rangle_{n}^{(L)} := \int_{\Delta}   x^{k}e^{-\alpha x}\rho^{(L)}_{n}(x)\, dx = \sum_{m=0}^{\infty}\frac{(-\alpha)^{m}}{m!}\mathcal{J}_{n,n}^{(L)}(k+m,\alpha,1),
\end{equation}
and 
\begin{equation}
\label{WFML1}
\langle \omega^{k}(x)\log\omega(x)\rangle_{n}^{(L)} := \int_{\Delta}  \omega^{k+1}(x)\log\omega(x) \rho^{(L)}_{n}(x)\, dx = \frac{d}{dk}\mathcal{J}_{n,n}^{(L)}(0,\alpha,k+1).
\end{equation}
From Eq. \eqref{GKL} one has that they are explicitly given by
 \begin{equation}
 \label{LLM2}
 \langle (\log x)^{k}\rangle_{n}^{(L)} = \frac{d^{k}}{ds^{k}}c_{0}\left( \alpha+s, 2,n, 1, \alpha,0 \right)\Bigg|_{s=0}
 \end{equation}
 \begin{equation}
 \label{EFML2}
 \langle x^{k}e^{-\alpha x}\rangle_{n}^{(L)} = \sum_{m=0}^{\infty}\frac{(-\alpha)^{m}}{m!}c_{0}(\alpha+k+m,2,n,1,\alpha,0)
 \end{equation}
 with
 \begin{align}
  \label{PMCS3}
  c_{0}\left( \alpha+k+m, 2,n, 1, \alpha,0 \right) &= \Gamma(\alpha +k+m +1) \binom{n+\alpha}{n}^{2}\nonumber \\
  & \hspace{-2cm} \times F_{A}^{(2)} \left(\begin{array}{cc}
  														 \alpha +k+m  +1 ; -n, -n& \\
  																										&; 1,1 \\
  														\alpha+1, \alpha+1 & \\
  														\end{array}\right),
  \end{align}
 and 
 \begin{align}
 \label{WFML2}
\hspace{-1cm} \langle \omega^{k}(x)\log\omega(x)\rangle_{n}^{(L)} &= \mathcal{J}^{(L)}_{n,n}(0,\alpha,k+1)\Bigg\{\frac{d}{dk}\log \left[c_{0}\left( (k+1)\alpha, 2, n, \frac{1}{k+1}, \alpha,0 \right)\right] \nonumber \\
& \hspace{4cm} -\frac{1}{k+1}-\alpha(1+\log(1+k)) \Bigg\}.
 \end{align}
 \end{enumerate}
\subsection{Hermite polynomials}
In this case the method gives the following result:
%

\begin{theorem}
Let $H_{m}(x)$ denote the Hermite polynomials orthogonal with respect to the weight function $\omega(x)=e^{-x^{2}}$ on $(-\infty,\infty)$ \cite{nikiforov1988,olver2010}. Then, the generalized Krein-like functionals of the Hermite polynomials defined by 
\begin{equation}
\label{GKLFH11}
\mathcal{J}_{\{m_{r} \}}^{(H)}(s,\beta)  = \int_{-\infty}^{\infty} [\omega(x)]^{\beta} x^{s}H_{m_{1}}(x)\cdots H_{m_{r}}(x)\, dx,\quad s\in \mathbb{Z}_{+}
\end{equation}
are given by
\begin{align}
\label{GKH}
\mathcal{J}^{(H)}_{ \{m_{r} \} } (s,\beta) &= \frac{\sqrt{\pi}}{\beta^{\frac{N+s+1}{2}}}\frac{(-1)^{\frac{1}{2}(M-N)}\prod_{i=1}^{r}m_{i}!}{\prod_{i=1}^{r}\left(\frac{1}{2}\right)_{\frac{m_{i}+\nu_{i}}{2}}}\,c_{0}\left(\frac{s+N}{2},r,\left\{\frac{m_{r}-\nu_{r}}{2} \right\},\frac{1}{\beta}, \left\{\nu_{r}-\frac{1}{2}\right\} \right),
\end{align}
where $M=m_{1}+\ldots+m_{r}$, $N=\nu_{1}+\ldots+\nu_{r}$, $\frac{s+N}{2}\in\mathbb{Z}_{+}$, $\nu=0(1)$ for even(odd) degree of the Hermite polynomials and the coefficient $c_{0}\left(\frac{s+N}{2},r,\left\{\frac{m_{r}-\nu_{r}}{2} \right\},\frac{1}{\beta}, \left\{\nu_{r}-\frac{1}{2}\right\} \right)$ is given by
\begin{align}
\label{SNHC2}
 c_{0}\left(\frac{s+N}{2}, r,\frac{1}{\beta},\left\{\frac{m_{r}-\nu_{r}}{2} \right\},\left\{\nu_{r}-\frac{1}{2}\right\},-\frac{1}{2}\right) &= \left(\frac{1}{2}\right)_{\frac{s+N}{2}} \binom{\frac{m_{1}+\nu_{1}-1}{2}}{\frac{m_{1}-\nu_{1}}{2}}\cdots \binom{\frac{m_{r}+\nu_{r}-1}{2}}{\frac{m_{r}-\nu_{r}}{2}} \nonumber \\
&\hspace{-4cm} \times F_{A}^{(r)} \left(\begin{array}{cc}
\frac{s+N+1}{2}; -\frac{m_{1}-\nu_{1}}{2}, \ldots, -\frac{m_{r}-\nu_{r}}{2} &\\
&; \frac{1}{\beta}, \ldots, \frac{1}{\beta} \\
\nu_{1}+ \frac{1}{2}, \ldots, \nu_{r}+ \frac{1}{2}& \\
\end{array}\right).
\end{align}
\end{theorem}
\begin{proof}
We first rewrite \eqref{GKLFH11} as
\begin{align}
\label{GKLFH2}
\mathcal{J}_{\{m_{r} \}}^{(H)}(s,\beta)  &= \int_{-\infty}^{\infty} e^{-\beta x^{2}}x^{s}H_{m_{1}}(x)\cdots H_{m_{r}}(x)\, dx \nonumber \\
&= \beta^{-\frac{1}{2}(s+1)} \int_{-\infty}^{\infty} e^{-y^{2}} y^{s} H_{m_{1}}\left(\frac{y}{\sqrt{\beta}}\right)\cdots H_{m_{r}}\left(\frac{y}{\sqrt{\beta}}\right)\, dy .
\end{align}
Then, we use the conversion formula \cite{chaggara2007}
	\begin{equation}
	\label{HLR1}
	H_{m_{i}}(x) = \frac{(-1)^{\frac{m_{i}-\nu_{i}}{2}}m_{i}!}{\left(\frac{1}{2}\right)_{\frac{m_{i}+\nu_{i}}{2}}}x^{\nu_{i}}L_{\frac{m_{i}-\nu_{i}}{2}}^{(\nu_{i}-\frac{1}{2})}(x^{2}),
	\end{equation}
with $\nu_{i} =0(1)$ for even(odd) degree of the Hermite polynomials, which allow us to convert the Hermite-dependent kernel of Eq. \eqref{HLR1} into the following Laguerre-dependent kernel:
\begin{align}
\label{HLR2}
\hspace{-1cm} y^{s}H_{m_{1}}\left(\frac{y}{\sqrt{\beta}}\right)\cdots H_{m_{r}}\left(\frac{y}{\sqrt{\beta}}\right) &= \frac{(-1)^{\frac{1}{2}(M-N)}\prod_{i=1}^{r}m_{i}!}{\prod_{i=1}^{r}\left(\frac{1}{2}\right)_{\frac{m_{i}+\nu_{i}}{2}}} \frac{y^{s+N}}{\beta^{\frac{N}{2}}}\nonumber \\
& \hspace{-1cm} \times  L_{\frac{m_{1}-\nu_{1}}{2}}^{(\nu_{1}-1/2)}\left(\frac{y^{2}}{\beta}\right)\cdots L_{\frac{m_{r}-\nu_{r}}{2}}^{(\nu_{r}-1/2)}\left(\frac{y^{2}}{\beta}\right),
\end{align}
Then, using the Laguerre-linearization formula \eqref{SNL} we obtain
\begin{align}
\label{HLR3}
y^{s}H_{m_{1}}\left(\frac{y}{\sqrt{\beta}}\right)\cdots H_{m_{r}}\left(\frac{y}{\sqrt{\beta}}\right) &= \frac{(-1)^{\frac{1}{2}(M-N)}\prod_{i=1}^{r}m_{i}!}{\prod_{i=1}^{r}\left(\frac{1}{2}\right)_{\frac{m_{i}+\nu_{i}}{2}}\beta^{\frac{N}{2}}} \nonumber \\
& \hspace{-2cm} \times \sum_{i=0}^{\infty} c_{i}\left(\frac{s+N}{2}, r,\frac{1}{\beta},\left\{\frac{m_{r}-\nu_{r}}{2} \right\},\left\{\nu_{r}-\frac{1}{2}\right\},-\frac{1}{2}\right)L_{i}^{(-\frac{1}{2})}(y^{2})\nonumber\\
&= \frac{(-1)^{\frac{1}{2}(M-N)}\prod_{i=1}^{r}m_{i}!}{\prod_{i=1}^{r}\left(\frac{1}{2}\right)_{\frac{m_{i}+\nu_{i}}{2}}\beta^{\frac{N}{2}}} \nonumber \\
&\hspace{-2cm} \times \sum_{i=0}^{\infty} c_{i}\left(\frac{s+N}{2}, r,\frac{1}{\beta},\left\{\frac{m_{r}-\nu_{r}}{2} \right\},\left\{\nu_{r}-\frac{1}{2}\right\},-\frac{1}{2}\right)\frac{1}{(-1)^{i}2^{2i}i!}H_{2i}(y)
\end{align} 
for $\frac{s+N}{2}\in\mathbb{Z}_{+}$ and with coefficients
\begin{align}
\label{SNLC3}
 c_{i}\left(\frac{s+N}{2}, r,\frac{1}{\beta},\left\{\frac{m_{r}-\nu_{r}}{2} \right\},\left\{\nu_{r}-\frac{1}{2}\right\},-\frac{1}{2}\right) &= \left(\frac{1}{2}\right)_{\frac{s+N}{2}} \binom{\frac{m_{1}+\nu_{1}-1}{2}}{\frac{m_{1}-\nu_{1}}{2}}\cdots \binom{\frac{m_{r}+\nu_{r}-1}{2}}{\frac{m_{r}-\nu_{r}}{2}} \nonumber \\
&\hspace{-4cm} \times F_{A}^{(r+1)} \left(\begin{array}{cc}
\frac{s+N+1}{2}; -\frac{m_{1}-\nu_{1}}{2}, \ldots, -\frac{m_{r}-\nu_{r}}{2}, -i &\\
&; \frac{1}{\beta}, \ldots, \frac{1}{\beta}, 1 \\
\nu_{1}+ \frac{1}{2}, \ldots, \nu_{r}+ \frac{1}{2}, \frac{1}{2} & \\
\end{array}\right)
\end{align}
where $F_{A}^{(r+1)}(x_{1},\ldots,x_{r})$ denotes the Lauricella function of type A.
Now, we take \eqref{HLR3} into Eq. \eqref{GKLFH2} obtaining
\begin{align}
\label{GKLFH3}
 \mathcal{J}_{\{m_{r} \}}^{(H)}(s,\beta)  &= \frac{(-1)^{\frac{1}{2}(M-N)}\prod_{i=1}^{r}m_{i}!}{\prod_{i=1}^{r}\left(\frac{1}{2}\right)_{\frac{m_{i}+\nu_{i}}{2}}\beta^{\frac{N+s+1}{2}}} \sum_{i=0}^{\infty} c_{i}\left(\frac{s+N}{2}, r,\frac{1}{\beta},\left\{\frac{m_{r}-\nu_{r}}{2} \right\},\left\{\nu_{r}-\frac{1}{2}\right\},-\frac{1}{2}\right)\nonumber \\ 
 &\times\frac{1}{(-1)^{i}2^{2i}i!} \int_{-\infty}^{\infty} e^{-y^{2}} H_{2i}(y) \, dy.
\end{align}
Using the orthogonality condition of the Hermite polynomials \cite{olver2010} one finally has the wanted expressions \eqref{GKH} and \eqref{SNHC2} of the theorem.
\end{proof}

\textbf{Applications.}
Let us now apply this theorem to the power, Krein-like, exponential and logarithmic moments of the Rakhmanov probability density $\rho_{n}^{(H)}(x)$ of Hermite polynomials $H_{n}(x)$ given by
$$\rho_{n}^{(H)}(x) = \omega(x)\,[H_{m}(x)]^2$$
\begin{enumerate}
	\item Power moments and Krein-like moments.
	\\
The power and Krein-like moments of the Hermite polynomials are given by
\begin{equation}
\label{HPM1}
\langle x^{s}\rangle_{n}^{(H)} = \int_{\Delta} x^{s}\rho^{(H)}_{n}(x)\, dx = \mathcal{J}_{n,n}^{(H)}(s,1),
\end{equation}
\begin{equation}
\label{HKM1}
\langle [\omega(x)]^{k}\rangle_{n}^{(H)} = \int_{\Delta}  [\omega(x)]^{k}\rho^{(H)}_{n}(x)\, dx = \mathcal{J}_{n,n}^{(H)}(0,k+1)
\end{equation}
respectively. From Eq. \eqref{GKH} one has that 
 \begin{align}
 \label{HPM2}
 \langle x^{s}\rangle_{n}^{(H)} &= \sqrt{\pi}\frac{(-1)^{\frac{M-N}{2}}(n!)^{2}}{(\frac{1}{2})_{\frac{n+\nu_{1}}{2}}(\frac{1}{2})_{\frac{n+\nu_{2}}{2}}}\nonumber\\ & \times c_{0}\left(\frac{s+N}{2},2,\left\{\frac{n-\nu_{1}}{2},\frac{n-\nu_{2}}{2}\right\},1, \left\{\nu_{1}-\frac{1}{2},\nu_{2}-\frac{1}{2}\right\} \right),
 \end{align}
 with $N=\nu_{1}+\nu_{2}$, 
 \begin{align}
 \label{HPMCS1}
 c_{0}\left( \frac{s+N}{2},2,\left\{\frac{n-\nu_{1}}{2},\frac{n-\nu_{2}}{2}\right\},1, \left\{\nu_{1}-\frac{1}{2},\nu_{2}-\frac{1}{2}\right\} \right) \nonumber \\ &\hspace{-8cm} =\left(\frac{1}{2}\right)_{\frac{s+N}{2}}\binom{\frac{n+\nu_{1}-1}{2}}{\frac{n-\nu_{1}}{2}} \binom{\frac{n+\nu_{2}-1}{2}}{\frac{n-\nu_{2}}{2}}  F_{A}^{(2)} \left(\begin{array}{cc}
 													\frac{s+N+1}{2} ; -\frac{n-\nu_{1}}{2}, -\frac{n-\nu_{2}}{2} &\\
 																										&; 1,1\\
 														\nu_{1}+\frac{1}{2}, \nu_{2}+\frac{1}{2}& \\
 														\end{array}\right),
 \end{align}
 and
 \begin{align}
 \label{HKM2}
 \langle [\omega(x)]^{k}\rangle_{n}^{(H)} &= \frac{\sqrt{\pi}}{(k+1)^{\frac{N+1}{2}}}\frac{(-1)^{\frac{M-N}{2}}(n!)^{2}}{(\frac{1}{2})_{\frac{n+\nu_{1}}{2}}(\frac{1}{2})_{\frac{n+\nu_{2}}{2}}}\nonumber \\
 & \times  c_{0}\left(\frac{N}{2},2,\left\{\frac{n-\nu_{1}}{2},\frac{n-\nu_{2}}{2}\right\},\frac{1}{k+1}, \left\{\nu_{1}-\frac{1}{2},\nu_{2}-\frac{1}{2}\right\}  \right)
 \end{align}
 with $M = 2n$, 
  \begin{align}
  \label{HPMCS2}
  c_{0}\left(\frac{N}{2},2,\left\{\frac{n-\nu_{1}}{2},\frac{n-\nu_{2}}{2}\right\},1, \left\{\nu_{1}-\frac{1}{2},\nu_{2}-\frac{1}{2}\right\}   \right)\nonumber \\ &\hspace{-8cm} =\left(\frac{1}{2}\right)_{\frac{N}{2}}\binom{\frac{n+\nu_{1}-1}{2}}{\frac{n-\nu_{1}}{2}} \binom{\frac{n+\nu_{2}-1}{2}}{\frac{n-\nu_{2}}{2}}\nonumber  F_{A}^{(2)} \left(\begin{array}{cc}
  														\frac{N+1}{2} ;-\frac{n-\nu_{1}}{2}, -\frac{n-\nu_{2}}{2} &\\
  																										&; \frac{1}{k+1}, \frac{1}{k+1}\\
  														\nu_{1}+\frac{1}{2}, \nu_{2}+\frac{1}{2}& \\
  														\end{array}\right),
  \end{align}
 for the power and Krein-like moments, respectively.

 \item Other exponential and logarithmic functionals.
 \\
 The logarithmic and exponential functionals of the Hermite polynomials are defined as
\begin{equation}
\label{HLM1}
\langle (\log x)^{k}\rangle_{n}^{(H)} := \int_{\Delta}  (\log x)^{k}\rho^{(H)}_{n}(x)\, dx = \frac{d^{k}}{ds^{k}}\mathcal{J}_{n,n}^{(H)}(s,1)\Bigg|_{s=0},
\end{equation}
\begin{equation}
\label{EFMH1}
\langle x^{k}e^{-\alpha x}\rangle_{n}^{(H)} := \int_{\Delta}   x^{k}e^{-\alpha x}\rho^{(H)}_{n}(x)\, dx = \sum_{m=0}^{\infty}\frac{(-\alpha)^{m}}{m!}\mathcal{J}_{n,n}^{(H)}(k+m,1),
\end{equation}
and 
\begin{equation}
\label{WFMH1}
\langle \omega^{k}(x)\log\omega(x)\rangle_{n}^{(H)} := \int_{\Delta}  \omega^{k+1}(x)\log\omega(x) \rho^{(H)}_{n}(x)\, dx = \frac{d}{dk}\mathcal{J}_{n,n}^{(H)}(0,k+1).
\end{equation}
From Eq. \eqref{GKH} one has that they are explicitly given by
\begin{align}
 \label{HLM2}
 \langle (\log x)^{k}\rangle_{n}^{(H)} &= \sqrt{\pi}\frac{(-1)^{\frac{M-N}{2}}(n!)^{2}}{(\frac{1}{2})_{\frac{n+\nu_{1}}{2}}(\frac{1}{2})_{\frac{n+\nu_{2}}{2}}}\nonumber \\
 & \times \frac{d^{k}}{ds^{k}} c_{0}\left(\frac{s+N}{2},2,\left\{\frac{n-\nu_{1}}{2},\frac{n-\nu_{2}}{2}\right\},1, \left\{\nu_{1}-\frac{1}{2},\nu_{2}-\frac{1}{2}\right\}  \right)\Bigg|_{s=0}.
 \end{align}
 \begin{align}
 \label{EFMH2}
 \langle x^{k}e^{-\alpha x}\rangle_{n}^{(H)} &=\sqrt{\pi}\frac{(-1)^{\frac{M-N}{2}}(n!)^{2}}{(\frac{1}{2})_{\frac{n+\nu_{1}}{2}}(\frac{1}{2})_{\frac{n+\nu_{2}}{2}}}\nonumber \\ & \times \sum_{m=0}^{\infty}\frac{(-\alpha)^{m}}{m!}\, c_{0}\left(\frac{k+m+N}{2},2,\left\{\frac{n-\nu_{1}}{2},\frac{n-\nu_{2}}{2}\right\},1, \left\{\nu_{1}-\frac{1}{2},\nu_{2}-\frac{1}{2}\right\}  \right)
 \end{align}
 with
\begin{align}
 \label{HPMCS3}
 c_{0}\left(\frac{k+m+N}{2},2,\left\{\frac{n-\nu_{1}}{2},\frac{n-\nu_{2}}{2}\right\},1, \left\{\nu_{1}-\frac{1}{2},\nu_{2}-\frac{1}{2}\right\} \right) &\nonumber \\ &\hspace{-8cm}=\left(\frac{1}{2}\right)_{\frac{k+m+N}{2}}\binom{\frac{n+\nu_{1}-1}{2}}{\frac{n-\nu_{1}}{2}} \binom{\frac{n+\nu_{2}-1}{2}}{\frac{n-\nu_{2}}{2}}\nonumber \\
 &\hspace{-8cm} \times F_{A}^{(2)} \left(\begin{array}{cc}
 														\frac{k+m+N+1}{2} ; -\frac{n-\nu_{1}}{2}, -\frac{n-\nu_{2}}{2} &\\
 																										&; 1,1\\
 														\nu_{1}+\frac{1}{2}, \nu_{2}+\frac{1}{2}& \\
 														\end{array}\right),
 \end{align}
 and 
 \begin{align}
 \label{WFMH2}
 \langle \omega^{k}(x)\log\omega(x)\rangle_{n}^{(H)} &= \mathcal{J}_{n,n}^{(H)}(0,k+1)\nonumber \\
 & \times \left\{\frac{d}{dk}\log \left[c_{0}\left(\frac{N}{2},2,\left\{\frac{n-\nu_{1}}{2},\frac{n-\nu_{2}}{2}\right\},1, \left\{\nu_{1}-\frac{1}{2},\nu_{2}-\frac{1}{2}\right\}   \right)\right]- \frac{N+1}{2(k+1)} \right\}.
 \end{align}
 \end{enumerate}

\subsection{Jacobi polynomials}

In this case the method gives the following result:
\begin{theorem}
Let $P_{m}^{(\alpha,\gamma)}(x)$ denote the Jacobi polynomials orthogonal with respect to the weight function $\omega_{\alpha,\gamma}(x)=(1-x)^{\alpha}(1+x)^{\gamma}$ on $(-1,1)$ \cite{nikiforov1988,olver2010}. Then, the generalized Krein-like functionals of the Jacobi polynomials defined by 
\begin{equation}
\label{GKLFJ11}
\mathcal{J}_{\{m_{r} \}}^{(J)}(s,\alpha,\gamma,\beta)  = \int_{-1}^{1} [\omega_{\alpha,\gamma}(x)]^{\beta} x^{s}P_{m_{1}}^{(\alpha,\gamma)}(x)\ldots P_{m_{r}}^{(\alpha,\gamma)}(x)\, dx, \quad \alpha,\gamma>-\frac{1}{\beta},\quad s\in\mathbb{Z}_{+}
\end{equation}
are given by
\begin{align}
\label{GKJ}
\mathcal{J}^{(J)}_{ \{m_{r} \} } (s,\alpha,\gamma,\beta) &= \frac{2^{\beta(\alpha+\gamma)+1}\Gamma(\alpha\beta+1)\Gamma(\gamma\beta+1) }{\Gamma(\beta(\alpha+\gamma)+2)}\tilde{c}_{0}(s,r,\{m_{r} \}, \alpha, \gamma,\beta\alpha,\beta\gamma)
\end{align}
 with
 \begin{align}
 \label{GKJC}
 \tilde{c}_{0}(s,r,\{m_{r} \}, \alpha, \gamma,\beta\alpha,\beta\gamma) &= (\beta\alpha+1)_{s} \binom{m_{1}+\alpha}{m_{1}}\cdots \binom{m_{r}+\alpha}{m_{r}} \frac{\beta(\alpha +\gamma )+1}{(\beta(\alpha +\gamma ) +1)_{s +1}}\nonumber\\
 & \hspace{-3.5cm}\times \sum_{j_{1},\ldots,j_{r}=0}^{m_{1},\ldots,m_{r}}\frac{(\beta\alpha +s +1)_{j_{1}+\ldots+j_{r}}}{(\beta(\alpha +\gamma) +s +2)_{j_{1}+\ldots+j_{r}}}\nonumber\\
 &\hspace{-3.5cm}\times \frac{(-m_{1})_{j_{1}}(m_{1}+\alpha+\gamma+1)_{j_{1}}\cdots(-m_{r})_{j_{r}}(m_{r}+\alpha+\gamma+1)_{j_{r}} }{(\alpha+1)_{j_{1}}\cdots (\alpha+1)_{j_{r}}}\frac{1}{j_{1}! \cdots j_{r}!}\nonumber\\
 & \hspace{-3.5cm}=(\beta\alpha+1)_{s} \binom{m_{1}+\alpha}{m_{1}}\cdots \binom{m_{r}+\alpha}{m_{r}} \frac{\beta(\alpha +\gamma )+1}{(\beta(\alpha +\gamma ) +1)_{s +1}}\nonumber\\
 &\hspace{-3.5cm}\times F_{1:1;\ldots;1}^{1:2;\ldots;2} \left(\begin{array}{cc}
  														  \beta\alpha+s+1: -m_{1}, \alpha+\gamma+m_{1}+1; \ldots; -m_{r}, \alpha+\gamma+m_{r}+1  &\\
  																										&; 1, \ldots, 1 \\
  														\beta(\alpha+\gamma)+s+2: \alpha+1; \ldots; \alpha+1 & \\
  														\end{array}\right)  ,\quad s\in\mathbb{Z}_{+}.
 \end{align}
 where $F_{2:1;\ldots;1}^{2:2;\ldots;2}(x_{1},\ldots,x_{r})$ denotes the $r$-variate Srivastava-Daoust function \cite{srivastava1988} defined as
 \begin{align}
 F_{1:1;\ldots;1}^{1:2;\ldots;2}\left( \begin{array}{cc}
 a_0^{(1)}:\,a_1^{(1)},a_1^{(2)}; \ldots;a_r^{(1)},a_r^{(2)} & \\
 &; x_1, \ldots, x_r\\
 b_0^{(1)}:\,b_1^{(1)}; \ldots;b_r^{(1)}& \\
 \end{array}\right) &=\nonumber\\
&\hspace{-8cm} =  \sum_{j_{1}, \ldots, j_r=0 }^{\infty} \frac{\left(a_0^{(1)}\right)_{j_1+\ldots+j_r}}{\left(b_0^{(1)}\right)_{j_1+\ldots+j_r}} \frac{\left(a_1^{(1)}\right)_{j_1}\left(a_1^{(2)}\right)_{j_1}\cdots\left(a_r^{(1)}\right)_{j_r}\left(a_r^{(2)}\right)_{j_r}}{\left(b_1^{(1)}\right)_{j_1}\left(b_r^{(1)}\right)_{j_r}}\frac{x_1^{j_1}x_2^{j_2}\cdots x_r^{j_r}}{j_1!j_2!\cdots j_r!}.
 \end{align}
\end{theorem}
\begin{proof} 
To begin with, we write Eq. \eqref{GKLFJ11} more transparently as
\begin{equation}
\label{GKLFJ2}
\mathcal{J}_{\{m_{r} \}}^{(J)}(s,\alpha,\gamma,\beta)  = \int_{-1}^{1} (1-x)^{\beta\alpha}(1+x)^{\beta\gamma} x^{s}P_{m_{1}}^{(\alpha,\gamma)}(x)\ldots P_{m_{r}}^{(\alpha,\gamma)}(x)\, dx.
\end{equation}	
Then, to obtain Eq. \eqref{GKJ}, we linearize the kernel product $x^{\mu}P_{m_1}^{(\alpha,\gamma)}(x)\cdots P_{m_r}^{(\alpha,\gamma)}(x)$ by using the method mentioned above, i.e.,
\begin{equation}
\label{SNJ1}
x^{\mu}P_{m_{1}}^{(\alpha_{1},\gamma_{1})}(x)\cdots P_{m_{r}}^{(\alpha_{r},\gamma_{r})}(x) = \sum_{i=0}^{\infty} \tilde{c}_{i}(\mu, r,t,\{m_{r} \},\{ \alpha_{r},\gamma_{r}\},\epsilon,\delta)P_{i}^{(\epsilon,\delta)}(x)
\end{equation} 
with coefficients
\begin{align}
\label{SNJC1}
c_{i}(\mu, r,t,\{m_{r} \},\{ \alpha_{r},\gamma_{r}\},\epsilon,\delta) &= (\epsilon+1)_{\mu} \binom{m_{1}+\alpha_{1}}{m_{1}}\cdots \binom{m_{r}+\alpha_{r}}{m_{r}} \frac{\epsilon +\delta +2 i+1}{(\epsilon +\delta+i +1)_{\mu +1}}\nonumber\\
& \hspace{-3.5cm}\times \sum_{j_{1},\ldots,j_{r}=0}^{m_{1},\ldots,m_{r}}\frac{(\epsilon +\mu +1)_{j_{1}+\ldots+j_{r}}}{(\epsilon +\delta+i +\mu +2)_{j_{1}+\ldots+j_{r}}}\, _2F_1(\epsilon +j_{1}+\ldots+j_{r}+\mu +1,-i;\epsilon +1;1)\nonumber\\
&\hspace{-3.5cm}\times \frac{(-m_{1})_{j_{1}}(m_{1}+\alpha_{1}+\gamma_{1}+1)_{j_{1}}\cdots(-m_{r})_{j_{r}}(m_{r}+\alpha_{r}+\gamma_{r}+1)_{j_{r}} }{(\alpha_{1}+1)_{j_{1}}\cdots (\alpha_{r}+1)_{j_{r}}}\frac{1}{j_{1}! \cdots j_{r}!},\quad \mu\in\mathbb{R}_{+}.
\end{align}
Now, we use Eq. \eqref{SNJ1} with $\alpha_{1}=\ldots=\alpha_{r}=\alpha$, $\gamma_{1}=\ldots=\gamma_{r}=\gamma$, $\mu=s$, $t=1$, $\epsilon =\beta\alpha$ and $\delta=\beta\gamma$ into Eq. \eqref{GKLFJ2} to obtain
\begin{equation}
\label{GKLFJ3}
\mathcal{J}_{\{m_{r} \}}^{(J)}(s,\alpha,\gamma,\beta)  = \sum_{i=0}^{s+m_{1}+\ldots+m_{r}} \tilde{c}_{i}(s,r,\{m_{r} \},\alpha,\gamma,\beta\alpha,\beta\gamma ) \int_{-1}^{1} (1-x)^{\beta\alpha}(1+x)^{\beta\gamma} P_{i}^{(\beta\alpha,\beta\gamma)}(x)\, dx.
\end{equation}
Finally, using the orthogonality property of the Jacobi polynomials \cite{olver2010} and that $_2F_1(a,0;c;1)=1$, Eq. \eqref{GKLFJ3} boils down to the wanted expressions \eqref{GKJ}-\eqref{GKJC} of the theorem.
\end{proof}

\subsection{Power and Krein-like moments}
\textbf{Applications.} Let us now apply this theorem to the power, Krein-like, exponential and logarithmic moments of the Rakhmanov probability density $\rho_{n}^{(J)}(x)$ of the Jacobi polynomials $P_{n}^{(\alpha,\gamma)}(x)$ given by
$$\rho_{n}^{(J)}(x) = \omega_{\alpha,\gamma}(x)\,[P_{m}^{(\alpha,\gamma)}(x)]^2$$
\begin{enumerate}
	\item Power and Krein-like moments.
	\\
The power and Krein-like moments of the Jacobi polynomials are given by
\begin{equation}
\label{JPM1}
\langle x^{s}\rangle_{n}^{(J)} = \int_{\Delta} x^{s}\rho^{(J)}_{n}(x)\, dx = \mathcal{J}_{n,n}^{(J)}(s,\alpha,\gamma,1),
\end{equation}
\begin{equation}
\label{JKM1}
\langle [\omega(x)]^{k}\rangle_{n}^{(J)} = \int_{\Delta}  [\omega(x)]^{k}\rho^{(J)}_{n}(x)\, dx = \mathcal{J}_{n,n}^{(J)}(0,\alpha,\gamma,k+1)
\end{equation}
respectively. From Eq. \eqref{GKJ} one has that 
 \begin{equation}
 \label{JPM2}
 \langle x^{s}\rangle_{n}^{(J)} = \frac{2^{\alpha+\gamma+1}\Gamma(\alpha+1)\Gamma(\gamma+1) }{\Gamma(\alpha+\gamma+2)}\,\tilde{c}_{0}(s,2,n, \alpha, \gamma,\alpha,\gamma)
 \end{equation}
 with
 \begin{align}
  \label{JPMCS1}
 \tilde{c}_{0}(s,2,n, \alpha, \gamma,\alpha,\gamma) &= (\alpha+1)_{s} \binom{n+\alpha}{n}^{2} \frac{\alpha +\gamma +1}{(\alpha +\gamma  +1)_{s +1}}\nonumber\\
 &\hspace{-3.5cm}\times F_{1:1;1}^{1:2;2} \left(\begin{array}{cc}
 \alpha+s+1: -n, \alpha+\gamma+n+1; -n, \alpha+\gamma+n+1  &\\
 &; 1, 1 \\
 \alpha+\gamma+s+2: \alpha+1; \alpha+1 & \\
 \end{array}\right)  
 \end{align}
 \begin{align}
 \label{JKM2}
 \langle [\omega(x)]^{k}\rangle_{n}^{(J)} &= \frac{2^{(k+1)(\alpha+\gamma)+1}\Gamma(\alpha(k+1)+1)\Gamma(\gamma(k+1)+1) }{\Gamma((k+1)(\alpha+\gamma)+2)}\nonumber \\
 & \times \tilde{c}_{0}(0,2,n, \alpha, \gamma,(k+1)\alpha,(k+1)\gamma)
 \end{align}
 with
 \begin{align}
   \label{JPMCS2}
  \tilde{c}_{0}(0,2,n, \alpha, \gamma,(k+1)\alpha,(k+1)\gamma) &=  \binom{n+\alpha}{n}^{2} \nonumber\\
  &\hspace{-5.cm}\times F_{1:1;1}^{1:2;2} \left(\begin{array}{cc}
  (k+1)\alpha+1: -n, \alpha+\gamma+n+1; -n, \alpha+\gamma+n+1  &\\
  &; 1, 1 \\
  (k+1)(\alpha+\gamma)+2: \alpha+1; \alpha+1 & \\
  \end{array}\right)
   \end{align}
 for the power and Krein-like moments, respectively.

 \item Other exponential and logarithmic functionals.
 \\
 The logarithmic and exponential functionals of the Laguerre polynomials are defined as
\begin{equation}
\label{JLM1}
\langle (\log x)^{k}\rangle_{n}^{(J)} = \int_{\Delta}  (\log x)^{k}\rho^{(J)}_{n}(x)\, dx = \frac{d^{k}}{ds^{k}}\mathcal{J}_{n,n}^{(J)}(s,\alpha,\gamma,1)\Bigg|_{s=0}.
\end{equation}
\begin{equation}
\label{EFMJ1}
\langle x^{k}e^{-\alpha x}\rangle_{n}^{(J)} := \int_{\Delta}   x^{k}e^{-\alpha x}\rho^{(J)}_{n}(x)\, dx = \sum_{m=0}^{\infty}\frac{(-\alpha)^{m}}{m!}\mathcal{J}_{n,n}^{(J)}(k+m,\alpha,\gamma,1),
\end{equation}
and 
\begin{equation}
\label{WFMJ1}
\langle \omega^{k}(x)\log\omega(x)\rangle_{n}^{(J)} := \int_{\Delta}  \omega^{k+1}(x)\log\omega(x) \rho^{(J)}_{n}(x)\, dx = \frac{d}{dk}\mathcal{J}_{n,n}^{(J)}(0,\alpha,\gamma,k+1).
\end{equation}
From Eq. \eqref{GKJ} one has that they are explicitly given by
 \begin{equation}
 \label{JLM2}
 \langle (\log x)^{k}\rangle_{n}^{(J)} =  \frac{2^{\alpha+\gamma+1}\Gamma(\alpha+1)\Gamma(\gamma+1) }{\Gamma(\alpha+\gamma+2)}\frac{d^{k}}{ds^{k}}\tilde{c}_{0}(s,2,n, \alpha, \gamma,\alpha,\gamma)\Bigg|_{s=0}
 \end{equation}
 \begin{equation}
 \label{EFMJ2}
 \langle x^{k}e^{-\alpha x}\rangle_{n}^{(J)} = \frac{2^{\alpha+\gamma+1}\Gamma(\alpha+1)\Gamma(\gamma+1) }{\Gamma(\alpha+\gamma+2)}\sum_{m=0}^{\infty}\frac{(-\alpha)^{m}}{m!}\,\tilde{c}_{0}(k+m,2,n, \alpha, \gamma,\alpha,\gamma), \quad k+m\neq 0
 \end{equation}
 with
\begin{align}
  \label{JPMCS3}
 \tilde{c}_{0}(k+m,2,n, \alpha, \gamma,\alpha,\gamma) &= (\alpha+1)_{k+m} \binom{n+\alpha}{n}^{2} \frac{\alpha +\gamma +1}{(\alpha +\gamma  +1)_{k+m +1}}\nonumber\\
 &\hspace{-3.5cm}\times F_{1:1;1}^{1:2;2} \left(\begin{array}{cc}
 \alpha+k+m+1: -n, \alpha+\gamma+n+1; -n, \alpha+\gamma+n+1  &\\
 &; 1, 1 \\
 \alpha+\gamma+k+m+2: \alpha+1; \alpha+1 & \\
 \end{array}\right),
  \end{align}
 and 
 \begin{align}
 \label{WFMJ2}
 \langle \omega^{k}(x)\log\omega(x)\rangle_{n}^{(J)} &= \mathcal{J}_{n,n}^{(J)}(0,k+1) \Bigg\{\frac{d}{dk}\log \left[\tilde{c}_{0}(0,2,n, \alpha, \gamma,(k+1)\alpha,(k+1)\gamma) \right] \nonumber \\
 & +\alpha H_{\alpha(k+1)}+\gamma H_{\gamma(k+1)}-(\alpha+\gamma)(H_{1+(k+1)(\alpha+\gamma)}-\log 2)   \Bigg\},
 \end{align}
where $\mathcal{H}_{m}$ denotes the harmonic numbers defined as $\mathcal{H}_{m} = \sum_{i=1}^{m}\frac{1}{i}$.
\end{enumerate}

\section{Differential-equation approach to Krein-like $2$-functionals of HOPs}

In this section we describe a method to compute the standard Krein-like functionals $\mathcal{J}_{m,n}(s,\beta)$, given by Eq. \eqref{RF2}, i.e.
\begin{equation}
	\label{fi1}
\mathcal{J}_{m,n}(s,\beta)= \int_{a}^{b} [\omega(x)]^{\beta}\,x^{s} p_{m}(x)p_{n}(x)\, dx
\end{equation}
where $\{p_{n}(x)\}$ represents a family of HOPs of degree $n$ on $[a,b]$ with respect to the weight function $\omega(x)$. This computation will be done  by using the second-order hypergeometric differential equation of the involved HOPs, obtaining expressions for $\mathcal{J}_{m,n}(s,\beta)$ directly in terms of the polynomial coefficients of the differential equation. Then, the utility of these general expressions is illustrated by applying them to the Laguerre, Hermite and Jacobi polynomials.

We start with the differential equation of the HOPs $\{p_{n}(x)\}$ which has the form
\begin{equation}
	\label{hyeq}
	\sigma(x)\,p_{n}''(x) +\tau(x)\,p_{n}'(x) +\lambda \,p_{n}(x) = 0,
\end{equation}
where $\sigma(x)$ and $\tau(x)$ are polynomials whose degree is not greater than $2$ and $1$, respectively, and $\lambda$ is a constant. These polynomial solutions satisfy the orthogonality condition
\begin{equation}
	\label{orthp}
	\int_{a}^{b} p_{k}(x)p_{l}(x)\rho(x)\, dx = h_{k}\delta_{k,l}\,.
\end{equation}
The constant $h_{n}$ can be shown \cite{nikiforov1988} as
\begin{equation}
	h_{n} = (-1)^{n}n!\,a_{n}B_{n}\int_{a}^{b}(\sigma(x))^{n}\omega(x)\, dx ,
\end{equation}
where $a_n$ is the coefficient of the leading term in the explicit expression of $\{p_{n}(x)\}$ (namely, $p_{n}(x) = a_n x^n + b_n x^{n-1} + ....$) and $B_{n}$ is the normalization constant appearing in the Rodrigues' formula
\begin{equation}
	\label{rof}
	p_{n}(x)= \frac{B_{n}}{\omega(x)}\frac{d^{n}}{dx^{n}}[(\sigma(x))^{n}\omega(x)].
\end{equation}
It is interesting to collect here the following relation between the constants $a_n$ and $B_{n}$ given by
\begin{equation}
	a_n = \frac{B_{n}}{\omega(x)}\frac{d^{n}}{dx^{n}}[(\sigma(x))^{n}\omega(x)] .
\end{equation}
To compute the wanted Krein functionals $\mathcal{J}_{mjk}(\beta)$ we consider the linearization relation of the form
\begin{equation}
	\label{EXP1}
	x^{s}p_{m}(x) = \sum_{k=0}^{s+m}c_{msk}p_{k}(x),
\end{equation}
where in general the linearization coefficients $c_{jmn}$ are given \cite{ruiz1997} by  
\begin{align}
	\label{eecof}
	c_{jmn} &= \frac{(-1)^{n}B_{n}B_{j}m!}{h_{n}} \sum_{r=r_{-}}^{r_{+}} \frac{A_{rj}}{(m-n+r)!}\int_{a}^{b} x^{m-n+r}(\sigma(x))^{n-r}\frac{d^{j-r}}{dx^{j-r}}[(\sigma(x))^{j}\rho(x)]\, dx,
\end{align}
where $r_{+}=\min(n,j)$, $r_{-}=\max(0,n-m)$, and the coefficients $A_{mn}$ come up in the generalized Rodrigues' formula \cite{nikiforov1988}
\begin{equation}
	\label{grof}
	\frac{d^{m}}{dx^{m}}p_{n}(x)=\frac{A_{mn}B_{n}}{(\sigma(x))^{m}\rho(x)}\frac{d^{n-m}}{dx^{n-m}}[(\sigma(x))^{n}\omega(x)],
\end{equation}
with
\begin{align}
	\label{Anm}
	A_{mn}&= (-1)^{m}\prod_{k=0}^{m-1} (\lambda_{n}-\lambda_{k}) = \frac{n!}{(n-m)!}\prod_{k=0}^{m-1}\left[\tau'+\frac{1}{2}(n+k-1)\sigma'' \right],\quad 1\leq m \leq n \nonumber \\
	A_{0n} &= 1.
\end{align}
Then, substituting \eqref{EXP1} into \eqref{fi1} one has 
\begin{equation}
	\label{fi12}
\mathcal{J}_{m,n}(s,\beta) = \sum_{k=0}^{s+m}c_{msk} \int_{a}^{b} [\omega(x)]^{\beta}\,p_{k}(x)p_{n}(x)\, dx,
\end{equation}
which, by a proper change of variable $x \to y=a\,x, a = a(\beta)$, can be transformed as
\begin{equation}
	\label{fi13}
\mathcal{J}_{m,n}(s,\beta) = \sum_{k=0}^{s+m}c'_{msk} \int_{a}^{b} \omega(y)\,p_{k}(ay)p_{n}(ay)\, dy.
\end{equation}
Now, one uses the linearization relation of the corresponding orthogonal polynomials, i.e.,
\begin{equation}
	\label{lin_rel_eq}
	p_{k}(x)p_{n}(x)=\sum_{r=0}^{k+n}d(k,n,r)p_{r}(x),
\end{equation}
where in general the coefficients are given by \cite{artes1998}
\begin{align}
	\label{coeff_lin_eq}
	d(n,m,r)&=\gamma_{n}\gamma_{m}\gamma_{r}\frac{(-1)^{r}A_{r,0}}{h_{r}}\sum_{j=j_{-}}^{j_{+}}\binom{r}{j}(-1)^{j}A_{n,j}\sum_{l=l_{-}}^{n-j}\binom{n-j}{l}(-1)^{l}A_{m,n+r-2j-l}\nonumber \\
	& \times \int_{a}^{b}f_{m}(x)\frac{d^{m-n-r+2j+l}}{dx^{m-n-r+2j+l}}\left\{\sigma^{2j+l-r}(x)\frac{d^{l}}{dx^{l}}(\sigma^{r-j}(x)) \right\}\, dx,
\end{align}
with $f_{n}(x)=\sigma^{n}(x)\omega(x)$, $\gamma_{i}$ represents the coefficient associated to the $n$-th power of the hypergeometric-type polynomial (i.e., $p_{n}(x)=\gamma_{n}x^{n}+\ldots$), $j_{-}=\max(0,r-m)$ and $j_{+}=\min(r,n)$. 
Then,
\begin{align}
	\label{fi14}
	\mathcal{J}_{m,n}(s,\beta) &= \sum_{n=0}^{s+m}c'_{msk} \sum_{r=0}^{k+n} d(k,n,r)\int_{a}^{b} \omega(y)\,p_{r}(ay)\, dy,\nonumber \\
	&\equiv \sum_{n=0}^{s+m}c'_{msk} \sum_{r=0}^{n+k} d(k,n,r)\,\mathcal{I}_{r}(a)
\end{align}
where the remaining integral, $\mathcal{I}_{n}(a)$, can be solved straightforwardly for each specific family of polynomials. \\
This expression formally gives the Krein-like functionals in terms of the linearization coefficients, $c_{msk}$ and $d(k,n,r)$, which are respectively  given by Eqs. \eqref{eecof} and \eqref{coeff_lin_eq} in terms of the coefficients $\sigma(x)$ and $\tau(x)$ of the differential equation satisfied by the hypergeometric polynomials involved.\\
Let us now use this method for generalized Krein-like integral functionals of the Laguerre, Hermite and Jacobi polynomials.

%
\subsection{Laguerre polynomials }

\begin{theorem}
	Let $L_{m}^{(\alpha)}(x)$ denote the Laguerre polynomials orthogonal with respect to the weight function $\omega_{\alpha}(x)=x^{\alpha}e^{-x}$ on $(0,\infty)$ \cite{nikiforov1988,olver2010}. Then, we found that the Krein-like functionals of Laguerre polynomials $L_{n}^{(\alpha)}(x)$ defined as
	\begin{equation}
		\label{GKL1}
		\mathcal{J}^{(L)}_{m,n}(s,\alpha,\beta):=\int_{0}^{\infty} x^{s}[\omega_{\alpha}^{(L)}(x)]^{\beta}L_{m}^{(\alpha)}(x)L_{n}^{(\alpha)}(x)\, dx, \quad \alpha >-\frac{s+1}{\beta},\quad s\in\mathbb{R}_{+}
	\end{equation}
	are given by
	\begin{align}
		\label{L1}
		\mathcal{J}_{m,n}^{(L)} &=\frac{1}{\beta^{\alpha+1}}\sum_{j=0}^{m+\gamma}c_{m\gamma j}\sum_{k=0}^{j+n}d(j,n,k)\frac{\Gamma(\alpha+1)(\alpha+1)_{k}}{k!}\left(\frac{\beta -1}{\beta }\right)^k,
	\end{align}
	where
	\begin{align}
		\label{cL}
		c_{m\gamma j} &=\frac{ (-1)^{j+m}\gamma ! \Gamma (\gamma +1)}{\Gamma (j+\alpha +1)}\sum_{l=\max(0,j-\gamma)}^{\min(j,m)} \binom{j}{l} \frac{\Gamma (l+\alpha +\gamma +1)}{(m-l)! (\gamma -j+l)! \Gamma (l-m+\gamma +1)}\\
		\label{dL}
		d(j,n,k)&=\frac{(-1)^{j+n}}{j!n!}\sum_{i=i_{-}}^{i_{+}}\binom{j}{i}\binom{n}{k-i}\sum_{r=0}^{r_{+}}\binom{j-i}{r}\binom{n+i-k}{r}r!\frac{k!}{(-1)^{k}}\nonumber \\
		&\hspace{-1cm} \times (k+\alpha+1)_{r}(k-j+r+1)_{j-i-r}(k-n+r+1)_{n+i-k-r}.
	\end{align}
	with $\gamma=s+\alpha(\beta-1)$, $i_{-}=\max(0,k-n)$, $i_{+}=\min(k,j)$ and $r_{+}=\min(j-i,n+i-k)$.
\end{theorem}
\begin{proof}
	\label{Lp}
	We begin from \eqref{GKL1} and perform the change of variable $y=\beta x$ to obtain
	\begin{equation}
		\label{Li1}
		\mathcal{J}_{m,n}^{(L)} = \frac{1}{\beta^{s+\alpha\beta+1}}\int_{0}^{\infty} y^{s+\alpha\beta-\alpha}y^{\alpha}e^{-y}L_{m}^{(\alpha)}\left(\frac{y}{\beta} \right)L_{n}^{(\alpha)}\left(\frac{y}{\beta} \right)\, dy,
	\end{equation}
	with a slight modification needed to apply the expansion \eqref{EXP1} one has
	\begin{equation}
		\label{Li2}
		\mathcal{J}_{m,n}^{(L)} = \frac{1}{\beta^{\alpha+1}}\int_{0}^{\infty} \left(\underbrace{\frac{y}{\beta}}_{z}\right)^{s+\alpha(\beta-1)}y^{\alpha}e^{-y}L_{m}^{(\alpha)}\left(\frac{y}{\beta} \right)L_{n}^{(\alpha)}\left(\frac{y}{\beta} \right)\, dy
	\end{equation}
	so that 
	\begin{equation}
		\label{Lexp}
		z^{\gamma} L_{m}^{(\alpha)}(z) = \sum_{j=0}^{\gamma+m}c_{m\gamma j}L_{j}^{(\alpha)}(z),
	\end{equation}
	with $\gamma=s+\alpha(\beta-1)$. The linearization coefficients $c_{m\gamma j}$ can be obtained directly from \eqref{eecof}, obtaining the values given by $\eqref{cL}$.  Then, substituting \eqref{Lexp} into \eqref{Li2} one has
	\begin{equation}
		\label{Li3}
		\mathcal{J}_{m,n}^{(L)} = \frac{1}{\beta^{\alpha+1}}\sum_{j=0}^{m+\gamma}c_{m\gamma j}\int_{0}^{\infty} L_{j}^{(\alpha)}\left(\frac{y}{\beta} \right)L_{n}^{(\alpha)}\left(\frac{y}{\beta} \right)y^{\alpha}e^{-y}\, dy.
	\end{equation}
Now, we use the linearization relation for the product of two Laguerre polynomials whose coefficients have been obtained following the method previously described \cite{ruiz1999}, i.e.,
\begin{equation}
\label{lin_2lag}
L_{n}^{(\alpha)}(x)L_{m}^{(\beta)}(x)=\sum_{k=0}^{n+m}d(n,m,k)L_{k}^{(\delta)}(x),
\end{equation}
with
\begin{align}
\label{coeff_2lag}
d(n,m,k)&=\frac{(-1)^{n+m}}{m!n!}\sum_{j=j_{-}}^{j_{+}}\binom{n}{j}\binom{m}{k-j}\sum_{r=0}^{r_{+}}\binom{n-j}{r}\binom{m+j-k}{r}r!\nonumber \\
&\hspace{-1.5cm} \times (k+\delta+1)_{r}(\delta-\alpha+k-n+r+1)_{n-j-r}(\delta-\beta+k-m+r+1)_{m+j-k-r}\frac{k!}{(-1)^{k}}.
\end{align}
Then, we obtain
\begin{equation}
\label{Li31}
\mathcal{J}_{m,n}^{(L)} = \frac{1}{\beta^{\alpha+1}}\sum_{j=0}^{m+\gamma}c_{m\gamma j}\sum_{k=0}^{j+n}d(j,n,k)\int_{0}^{\infty} L_{k}^{(\delta)}\left(\frac{y}{\beta} \right)y^{\alpha}e^{-y}\, dy
\end{equation}
which, for $\delta=\alpha$ and taking into account that $L_{0}^{(\alpha)}(y/\beta)=1$, can also be written as
\begin{equation}
\label{Li32}
\mathcal{J}_{m,n}^{(L)} = \frac{1}{\beta^{\alpha+1}}\sum_{j=0}^{m+\gamma}c_{m\gamma j}\sum_{k=0}^{j+n}d(j,n,k)\int_{0}^{\infty} L_{k}^{(\alpha)}\left(\frac{y}{\beta} \right)L_{0}^{(\alpha)}\left(\frac{y}{\beta} \right)y^{\alpha}e^{-y}\, dy,
\end{equation}
	which allows as to use the following integral \cite{ryzhik2007}
	\begin{align}
		\label{Lint}
		\int_{0}^{\infty} L_{k}^{(\alpha)}\left(\frac{y}{\beta} \right)L_{n}^{(\alpha)}\left(\frac{y}{\beta} \right)y^{\alpha}e^{-y}\, dy &= \frac{\Gamma(\alpha+1)(\alpha+1)_{k}(\alpha+1)_{n}}{k!n!}\nonumber \\
		& \times \sum_{l=0}^{k}\sum_{i=0}^{n}\frac{(-k)_{l}(-n)_{i}(l+\alpha+1)_{i} }{(\alpha+1)_{i}l!i!}\beta^{-l-i}.
	\end{align}
Indeed, this integral get simplified when $n=0$, i.e.,
	\begin{align}
\label{Lint2}
\int_{0}^{\infty} L_{k}^{(\alpha)}\left(\frac{y}{\beta} \right)L_{0}^{(\alpha)}\left(\frac{y}{\beta} \right)y^{\alpha}e^{-y}\, dy &= \frac{\Gamma (\alpha +1) (\alpha +1)_k}{k!}\left(\frac{\beta -1}{\beta }\right)^k,
\end{align}
and substituting \eqref{Lint2} into \eqref{Li32} one finally has
\begin{align}
\label{Li33}
\mathcal{J}_{m,n}^{(L)} &= \frac{1}{\beta^{\alpha+1}}\sum_{j=0}^{m+\gamma}c_{m\gamma j}\sum_{k=0}^{j+n}d(j,n,k)\frac{\Gamma(\alpha+1)(\alpha+1)_{k}}{k!} \sum_{l=0}^{k}\frac{(-k)_{l} }{l!}\beta^{-l}\nonumber \\
&=\frac{1}{\beta^{\alpha+1}}\sum_{j=0}^{m+\gamma}c_{m\gamma j}\sum_{k=0}^{j+n}d(j,n,k)\frac{\Gamma(\alpha+1)(\alpha+1)_{k}}{k!}\left(\frac{\beta -1}{\beta }\right)^k.
\end{align}

\end{proof} 

\subsection{Hermite polynomials }

\begin{theorem}
	Let $H_{m}(x)$ denote the Hermite polynomials orthogonal with respect to the weight function $\omega(x)=e^{-x^{2}}$ on $(-\infty,\infty)$. Then, the (standard) Krein-like functionals of the Hermite polynomials defined by
	\begin{equation}
		\label{gKH1}
		\mathcal{J}^{(H)}_{m,n}(s,\beta):=\int_{-\infty}^{\infty} x^{s}[\omega^{(H)}(x)]^{\beta}H_{m}(x)H_{n}(x)\, dx, \quad s\in \mathbb{Z}_{+}
	\end{equation} are given by
	\begin{align}
		\label{H1}
		\mathcal{J}_{m,n}^{(H)} &= \frac{1}{\sqrt{\beta}}\sum_{j=0}^{m+s}c_{msj}\sum_{k=0}^{\min(J,n)}\binom{J}{k}\binom{n}{k}\frac{2^{-2 j- k+m+n+s} \pi\, k! }{\Gamma \left(\frac{1+2(j+k)-(m+n)-s}{2} \right)}\left(\frac{\beta -1}{\beta }\right)^{\frac{-2( j+ k)+m+n+s}{2}},
	\end{align}
	with $J=m+s-2j$ and where
	\begin{equation}
		\label{cH}
		c_{msj} =\frac{2^{m-s} m! s! }{(m+s-2 j)!}\sum_{k=\max(0,m-2j)}^{\min(m+s-2j,m)} \binom{m+s-2 j}{k}\frac{1}{2^k (m-k)! (k+j-m)!}
	\end{equation}
\end{theorem}
\begin{proof}
	\label{Hp}
	We start by doing the change of variable $y=\sqrt{\beta}x$ in \eqref{gKH1}, obtaining
	\begin{equation}
		\label{Hi1}
		\mathcal{J}_{m,n}^{(H)} = \frac{1}{\beta^{\frac{s+1}{2}}}\int_{-\infty}^{\infty} y^{s}e^{-y^{2}}H_{m}\left( \frac{y}{\sqrt{\beta}}\right)H_{n}\left( \frac{y}{\sqrt{\beta}}\right)\, dy.
	\end{equation}
	In order to have an expansion of the type \eqref{EXP1}, we slightly rewrite the integral as
	\begin{equation}
		\label{Hi2}
		\mathcal{J}_{m,n}^{(H)} = \frac{1}{\sqrt{\beta}}\int_{-\infty}^{\infty} \left(\underbrace{\frac{y}{\sqrt{\beta}}}_{z}\right)^{s}e^{-y^{2}}H_{m}\left( \frac{y}{\sqrt{\beta}}\right)H_{n}\left( \frac{y}{\sqrt{\beta}}\right)\, dy
	\end{equation}
	so as to have
	\begin{equation}
		\label{exp2}
		z^{s}H_{m}(z) = \sum_{j=0}^{m+s}c_{msj}H_{m+s-2j}(z),
	\end{equation}
	The linearization coefficients $c_{msj}$ can be obtained directly from \eqref{eecof}, with the corresponding values given by $\eqref{cH}$. Thus, the integral turns out to be
	\begin{equation}
		\label{Hi3}
		\mathcal{J}_{m,n}^{(H)} = \frac{1}{\sqrt{\beta}}\sum_{j=0}^{m+s}c_{msj}\int_{-\infty}^{\infty} H_{m+s-2j}\left( \frac{y}{\sqrt{\beta}}\right)H_{n}\left( \frac{y}{\sqrt{\beta}}\right)e^{-y^{2}}\, dy.
	\end{equation}
	Now, we apply the linearization formula for the product of two Hermite polynomials whose coefficients are obtained again following the procedure described above \cite{artes1998}
	\begin{equation}
		\label{linforH}
		H_{J}\left( \frac{y}{\sqrt{\beta}}\right)H_{n}\left( \frac{y}{\sqrt{\beta}}\right) =\sum_{k=0}^{\min(J,n)}\binom{J}{k}\binom{n}{k}2^{k}k! \, H_{J+n-2k}\left( \frac{y}{\sqrt{\beta}}\right) 
	\end{equation}
	with $J=m+s-2j$ to have
	\begin{equation}
		\label{Hi4}
		\mathcal{J}_{m,n}^{(H)} = \frac{1}{\sqrt{\beta}}\sum_{j=0}^{m+s}c_{msj}\sum_{k=0}^{\min(J,n)}\binom{J}{k}\binom{n}{k}2^{k}k!\int_{-\infty}^{\infty}H_{J+n-2k}\left( \frac{y}{\sqrt{\beta}}\right) e^{-y^{2}}\, dy
	\end{equation}
	where the integral is given by \cite{ryzhik2007}
\begin{align}
\label{Hii}
	\int_{-\infty}^{\infty}H_{J+n-2k}\left( \frac{y}{\sqrt{\beta}}\right) e^{-y^{2}}\, dy &=\frac{2^{-2 j-2 k+m+n+s} \pi }{\Gamma \left(\frac{1+2(j+k)-(m+n)-s}{2} \right)}\left(\frac{\beta -1}{\beta }\right)^{\frac{-2( j+ k)+m+n+s}{2}}.
\end{align}
 Thus, taking \eqref{Hii} into \eqref{Hi4} one has the wanted \eqref{H1} of theorem.
\end{proof}

\subsection{Jacobi polynomials }

\begin{theorem}
	Let $P_{m}^{(\alpha,\gamma)}(x)$ denote the Jacobi polynomials orthogonal with respect to the weight function $\omega_{\alpha,\gamma}(x)=(1-x)^{\alpha}(1+x)^{\gamma}$ on $(-1,1)$. Then, the (standard) Krein-like functionals of the Jacobi polynomials defined by 
	\begin{equation}
		\label{GKJ1}
		\mathcal{J}^{(J)}_{m,n}(s,\alpha,\gamma,\beta):=\int_{-1}^{1} x^{s}[\omega_{\alpha,\gamma}^{(J)}(x)]^{\beta}P_{m}^{(\alpha,\gamma)}(x)P_{n}^{(\alpha,\gamma)}(x)\, dx, \quad s\in\mathbb{Z}_{+}
	\end{equation}
	can be expressed as
	\begin{align}
		\label{P1}
		\mathcal{J}_{m,n}^{(P)} &= \sum_{k=0}^{n}\sum_{j=0}^{m}\sum_{i=0}^{k+s}A_{k}(\alpha,\gamma)A_{j}(\alpha,\gamma)c_{ksi}\frac{2^{\Lambda+\Delta+1}\Gamma(\Lambda+j+1)\Gamma(\Delta+j+1)}{j!(\Lambda+\Delta+2j+1)\Gamma(\Lambda+\Delta+j+1)}\delta_{i,j},
	\end{align}
	where
	\begin{align}
		\label{Ai}
		A_{i}(\alpha,\gamma) &=  \frac{(-1)^{i-k} (\Delta +2 k+\Lambda +1) \Gamma (i+\gamma +1) \Gamma (i+\Lambda +1) \Gamma (k+\Delta +\Lambda +1) \Gamma (i+k+\alpha +\gamma +1)}{(i-k)! \Gamma (k+\gamma +1) \Gamma (k+\Lambda +1) \Gamma (i+\alpha +\gamma +1) \Gamma (i+k+\Delta +\Lambda +2)} \nonumber\\
		&\times \, _3F_2(k-i,-\alpha -i,\Delta +k+1;-i-\Lambda ,\gamma +k+1;1),\quad i=k,j
	\end{align}
	and
	\begin{align}
		\label{cP}
		c_{ksi} &= \frac{(-1)^{k+s-i}2^{i}s!(2i+\Lambda+\Delta+1)\Gamma(i+\Lambda+\Delta+1)\Gamma(k+\Lambda+1)\Gamma(k+\Delta+1) }{\Gamma(i+\Lambda+1)\Gamma(i+\Delta+1)\Gamma(k+\Lambda+\Delta+1) }\nonumber \\
		&\times \sum_{r=r_{-}}^{r_{+}} \binom{i}{r} \frac{\Gamma(k+r+\Lambda+\Delta+1)}{2^{r}(s-i+r)!\Gamma(2i+k-r+\Lambda+\Delta+2) }\nonumber \\
		& \times \sum_{l=0}^{k-r}\frac{(-1)^{l}\Gamma(i+k-r+\Lambda-l+1)\Gamma(i+\Delta+l+1) }{l!(k-r-l)!\Gamma(k+\Lambda-l+1)\Gamma(r+\Delta+l+1)}\nonumber\\
		& \times {}_2F_1\left(i-s-r,i+\Delta+l+1;2i+k-r+\Lambda+\Delta+2;2\right),
	\end{align}
	with $\Lambda = \alpha\beta$, $\Delta=\gamma\beta$, $r_{+}=\min(i,k)$ and $r_{-}=\max(0,i-s)$.
\end{theorem}
\begin{proof}
	\label{Pp}
	We start from \eqref{GKJ1} 
	\begin{equation}
		\label{Pi1}
		\mathcal{J}_{m,n}^{(P)} = \int_{-1}^{1}x^{s}(1-x)^{\alpha\beta}(1+x)^{\gamma\beta}P_{m}^{(\alpha,\gamma)}(x)P_{n}^{(\alpha,\gamma)}(x)\, dx
	\end{equation}
	and consider the expansion of the Jacobi polynomials in terms of the parameters, i.e.,
	\begin{align}
		\label{Pexp}
		P_{i}^{(\alpha,\gamma)}(x) &= \sum_{k=0}^{i}\frac{(-1)^{i-k} (\Delta +2 k+\Lambda +1) \Gamma (i+\gamma +1) \Gamma (i+\Lambda +1) \Gamma (k+\Delta +\Lambda +1) \Gamma (i+k+\alpha +\gamma +1)}{(i-k)! \Gamma (k+\gamma +1) \Gamma (k+\Lambda +1) \Gamma (i+\alpha +\gamma +1) \Gamma (i+k+\Delta +\Lambda +2)} \nonumber\\
		&\times \, _3F_2(k-i,-\alpha -i,\Delta +k+1;-i-\Lambda ,\gamma +k+1;1)P_{k}^{(\alpha\beta,\gamma\beta)}(x).
	\end{align}
	Then, \eqref{Pi1} turns out to be
	\begin{equation}
		\label{Pi2}
		\mathcal{J}_{m,n}^{(P)} = \sum_{k=0}^{n}\sum_{j=0}^{m}A_{k}(\alpha,\gamma)A_{j}(\alpha,\gamma)\int_{-1}^{1}x^{s}(1-x)^{\Lambda}(1+x)^{\Delta}P_{k}^{(\Lambda,\Delta)}(x)P_{j}^{(\Lambda,\Delta)}(x)\, dx
	\end{equation}
	with $\Lambda = \alpha\beta$ and $\Delta=\gamma\beta$. Thus, we use the following expansion of \eqref{EXP1}-type 
	\begin{equation}
		\label{Pexp1}
		x^{s}P_{k}^{(\Lambda,\Delta)}(x) = \sum_{i=0}^{k+s}c_{ksi}P_{i}^{(\Lambda,\Delta)}(x),
	\end{equation}
	(The linearization coefficients $c_{ksi}$ can be obtained directly from \eqref{eecof}, with the values given by $\eqref{cP}$.) to obtain
	\begin{equation}
		\label{Pi3}
		\mathcal{J}_{m,n}^{(P)} = \sum_{k=0}^{n}\sum_{j=0}^{m}A_{k}(\alpha,\gamma)A_{j}(\alpha,\gamma)\sum_{i=0}^{k+s}c_{ksi}\int_{-1}^{1}P_{i}^{(\Lambda,\Delta)}(x)P_{j}^{(\Lambda,\Delta)}(x)(1-x)^{\Lambda}(1+x)^{\Delta}\, dx
	\end{equation}
	where clearly the integral is the orthogonality relation of the Jacobi polynomials
	\begin{equation}
		\label{Portho}
		\int_{-1}^{1}P_{i}^{(\Lambda,\Delta)}(x)P_{j}^{(\Lambda,\Delta)}(x)(1-x)^{\Lambda}(1+x)^{\Delta}\, dx = \frac{2^{\Lambda+\Delta+1}\Gamma(\Lambda+j+1)\Gamma(\Delta+j+1)}{j!(\Lambda+\Delta+2j+1)\Gamma(\Lambda+\Delta+j+1)}\delta_{i,j}.
	\end{equation}
	Now, by taking \eqref{Portho} into \eqref{Pi3} one finally has the wanted expression \eqref{P1} of the theorem, where the coefficients $A_{i}(\alpha,\gamma)$ are given in \eqref{Ai}.
\end{proof}

\section{Algebraic approach to Krein-like $2$-functionals of HOPs }
In this section we compute the functionals $\mathcal{J}_{m,n}(s,\beta)$, given by Eq. \eqref{RF2}, for the Laguerre, Hermite and Jacobi polynomials by means of various characterizations of these HOPs, obtaining explicit expressions in terms of the degrees and the characteristic parameters of the associated weight functions. Then, these expressions are applied for the evaluation of the moments of power, Krein-like and logarithmic types.

\subsection{Laguerre polynomials}

\begin{theorem}
Let $L_{m}^{(\alpha)}(x)$ denote the Laguerre polynomials orthogonal with respect to the weight function $\omega_{\alpha}(x)=x^{\alpha}e^{-x}$ on $(0,\infty)$ \cite{nikiforov1988,olver2010}. Then, the (standard) Krein-like functionals of the Laguerre polynomials defined by 	
\begin{equation}
	\label{GKL1}
	\mathcal{J}^{(L)}_{m,n}(s,\alpha,\beta):=\int_{0}^{\infty} x^{s}[\omega_{\alpha}^{(L)}(x)]^{\beta}L_{m}^{(\alpha)}(x)L_{n}^{(\alpha)}(x)\, dx, \quad \alpha >-\frac{s+1}{\beta},\quad s\in\mathbb{R}_{+}
	\end{equation}
	are given by
	\begin{align}
	\label{GKL2}
	\hspace{-2cm} \mathcal{J}_{m,n}^{(L)}(s,\alpha,\beta)&  =(-2)^{n+m}\, \beta^{-1-s-\alpha
		\beta}\, \sum_{k=|n-m|}^{n+m} \left\{ \frac{(-1)^k k!}{2^k
		(m+n-k)!}\,
	\sum_{j=\max\{0,n-k,m-k\}}^{\left[\frac{m+n-k+1}{2}\right]}
	\frac{\left(\frac{k-m-n}{2}\right)_j
		\left(\frac{k-m-n+1}{2}\right)_j
		(\alpha+k+1)_j}{j! \, \Gamma(k-n+1+j)\,\Gamma(k-m+1+j) } \right\} \nonumber \\
	&\times \sum_{t=0}^{k} \frac{(-1)^t}{t!} {k+\alpha \choose k-t}
	\beta^{-t} \Gamma(\alpha \beta+s+t+1)\nonumber \\
	&\hspace{-2cm} =\frac{(-2)^{n+m}\, \beta^{-1-s-\alpha \beta}\,\Gamma(\alpha
		\beta+s+1)}{\Gamma(\alpha+1)} \sum_{k=|n-m|}^{n+m} \left\{
	\frac{(-1)^k \Gamma(\alpha+k+1)\, _2F_1\left(-k,\alpha \beta+s+1;
		\alpha +1;1/\beta \right)}{2^k (m+n-k)!}\, \right.\nonumber \\
	& \times
	\left. \sum_{j=\max\{0,n-k,m-k\}}^{\left[\frac{m+n-k+1}{2}\right]}
	\frac{\left(\frac{k-m-n}{2}\right)_j
		\left(\frac{k-m-n+1}{2}\right)_j
		(\alpha+k+1)_j}{j! \, \Gamma(k-n+1+j)\,\Gamma(k-m+1+j) } \right\}.
	\end{align}
	In case that $m=n$, this expression reduces to
	\begin{align}
	\label{GKL3}
	\mathcal{J}_{n,n}^{(L)}(s,\alpha,\beta)&=2^{2n}\, \beta^{-1-s-\alpha \beta}\,
	\sum_{k=0}^{2n} \left\{ \frac{(-1)^k k!}{2^k (2n-k)!}\,
	\sum_{j=\max\{0,n-k\}}^{\left[n-k/2+1/2\right]}
	\frac{\left(\frac{k}{2}-n\right)_j \left(\frac{k+1}{2}-n\right)_j
		(k+\alpha+1)_j}{j!\, \Gamma^2(k-n+1+j) } \right\} \nonumber \\
	&\times \sum_{t=0}^{k} \frac{(-1)^t}{t!} {k+\alpha \choose k-t}
	\beta^{-t} \Gamma(\alpha \beta+s+t+1) \nonumber \\
	&=\frac{2^{2n}\, \beta^{-1-s-\alpha \beta}\,\Gamma(\alpha
		\beta+s+1)}{\Gamma(\alpha+1)} \sum_{k=0}^{2n} \left\{ \frac{(-1)^k
		\Gamma(\alpha+k+1)\, _2F_1\left(-k,\alpha \beta +s+1;
		\alpha+1;1/\beta \right)}{2^k (2n-k)!}\, \right. \nonumber \\
	& \times
	\left. \sum_{j=\max\{0,n-k\}}^{\left[n-k/2+1/2\right]}
	\frac{\left(\frac{k}{2}-n\right)_j \left(\frac{k+1}{2}-n\right)_j
		(k+\alpha+1)_j}{j! \,\Gamma^2(k-n+1+j) } \right\}.
	\end{align}
\end{theorem}

\begin{proof}
	The keys to establish this result are the explicit expression of classical Laguerre polynomials \cite{olver2010} and the linearization formula  for the product of the Laguerre orthogonal polynomials \cite{ruiz1997,ruiz1999}. These formulae are
	\begin{align*}
	L_{n}^{(\alpha)}(x)&= \sum_{i=0}^n\frac{(-1)^i}{i!} {n+\alpha \choose n-i} x^i, \\
	 L_{n}^{(\alpha)}(x)L_{m}^{(\alpha)}(x)&= (-2)^{n+m} \sum_{k=|n-m|}^{n+m} \frac{(-1)^k k!}{2^k (m+n-k)!}\nonumber\\
	 & \times \sum_{j=\max \{0,n-k,m-k\}}^{[\frac{m+n-k+1}{2}]} \frac{ \left(\frac{k-m-n}{2} \right)_j \left(\frac{k-m-n+1}{2} \right)_j (\alpha+k+1)_j}{j! \, \Gamma(k-n+1+j) \, \Gamma(k-m+1+j)} L_{n}^{(\alpha)}(x).
	\end{align*}
	It only remains to combine the above expressions adequately and use the following integral
	$$
	\int_0^{\infty} x^{n+\alpha \beta} e^{-\beta x}dx=\frac{\Gamma (n+\alpha \beta+1)}{ \beta^{n+1+\alpha \beta}}. $$
	Then, after some computations, we obtain the result for $\mathcal{J}_{m,n}^{(L)}(s,\beta).$
\end{proof}

\textbf{Applications.} Let us now apply this theorem to compute the power, Krein-like and logarithmic moments of the Rakhmanov probability density $\rho_{n}^{(L)}(x)$ of the Laguerre polynomials $L_{n}^{(\alpha)}(x)$.

\begin{enumerate}
	\item Power moments.
	From \eqref{GKL2}, with $m=n$ and $\beta =1$, one obtains the analytical expression of the corresponding power moments of the Laguerre polynomials as
	\begin{align}
	\label{PML1}
	\langle x^{s}\rangle_{n}^{(L)} &=\mathcal{J}_{n,n}^{(L)}(s,\alpha,1)\nonumber \\
	&=\frac{2^{2n}\, \Gamma(\alpha+s+1 )}{\Gamma(\alpha+1)}
	\sum_{k=0}^{2n} \left\{ \frac{(-1)^k \Gamma(\alpha+k+1)\,
		_2F_1\left(-k,\alpha+s+1 ;
		\alpha+1;1 \right)}{2^k (2n-k)!}\, \right.  \nonumber \\
	& \times
	\left. \sum_{j=\max\{0,n-k\}}^{\left[n-k/2+1/2\right]}
	\frac{\left(\frac{k}{2}-n\right)_j \left(\frac{k+1}{2}-n\right)_j
		(\alpha+k+1)_j}{j! \, \Gamma^2(k-n+1+j) } \right\}.
	\end{align}
	
	\item Krein-like moments. The Krein-like moments of Laguerre polynomials, obtained from \eqref{GKL2} by making $m=n$, $s=0$ and $\beta =k+1$, are given by
	\begin{align}
	\label{KML1}
	\langle [\omega(x)]^{k}\rangle_{n}^{(L)}  & =  \mathcal{J}^{(L)}_{n,n}(0,\alpha,k+1) \nonumber \\
	&=\frac{2^{2n}\,\Gamma(\alpha
		(k+1)+1)}{ (k+1)^{\alpha (k+1)+1}\,\Gamma(\alpha+1)} \sum_{i=0}^{2n} \left\{ \frac{(-1)^i
		\Gamma(\alpha+i+1)\, _2F_1\left(-i,\alpha (k+1)+1;
		\alpha+1;\frac{1}{k+1} \right)}{2^i (2n-i)!}\, \right. \nonumber \\
	& \times
	\left. \sum_{j=\max\{0,n-i\}}^{\left[n-i/2+1/2\right]}
	\frac{\left(\frac{i}{2}-n\right)_j \left(\frac{i+1}{2}-n\right)_j
		(\alpha+i+1)_j}{j!\, \Gamma^2(i-n+1+j) } \right\}.
	\end{align}
	
	\item Logarithmic moments. The logarithmic moments of the Laguerre polynomials are obtained from the power moments as
	\begin{align}
	\label{LML1}
	\langle (\log x)^{k}\rangle_{n}^{(L)} &= \frac{d^{k}}{ds^{k}} \langle x^{s}\rangle_{n}^{(L)} \Bigg|_{s=0}.
	\end{align}
\end{enumerate}

\textbf{An example}

An advantage of this analytical expression is that it is very efficient computationally when we want to calculate the exact value of the integral. For example, choosing $\alpha=4, $ we compute  easily
\begin{align}
\mathcal{J}^{(L)}_{7,15}(2,3)= \int_{0}^{\infty} x^{2}[\omega_{4}^{(L)}(x)]^{3}L_{7}^{(4)}(x)L_{15}^{(4)}(x)\, dx &= \nonumber \\ &\hspace{-4cm} =\frac{10908801561641984000}{68630377364883}\approx 158950.04487071103773.
\end{align}
The computations have been made with \ma version 11. The mean timing was 0.00178125 seconds over 1000 trials. The straight evaluation of this integral using \ma needs 0.209031 seconds over the same number of trials. The timing according to \ma  reports CPU time used, in this case with a standard computer. Moreover,  the mean absolute time over 1000 trials, i.e. the total time to do the computation averaged over 1000 trials, was 0.00175197 using (\ref{GKL2}) and 0.215385 seconds using  straight evaluation of the integral via \ma. In this example, (\ref{GKL2}) was more than a hundred times faster than the evaluation of the integral via this software.

\subsection{Hermite polynomials }

\begin{theorem}
Let $H_{m}(x)$ denote the Hermite polynomials orthogonal with respect to the weight function $\omega(x)=e^{-x^{2}}$ on $(-\infty,\infty)$. Then, the (standard) Krein-like functionals of the Hermite polynomials defined by
\begin{equation}
\label{gKH1}
\mathcal{J}^{(H)}_{m,n}(s,\beta):=\int_{-\infty}^{\infty} x^{s}[\omega^{(H)}(x)]^{\beta}H_{m}(x)H_{n}(x)\, dx,\quad s\in\mathbb{Z}_{+}
\end{equation}
have the following analytical expression
\begin{align}
\label{GKH2}
\mathcal{J}^{(H)}_{m,n}(s,\beta) &= \frac{2^{m+n}}{\beta^{(m+n+s+1)/2}}\, \sum_{k=0}^{\min(m,n)}{ m \choose k} { n \choose k}  k! \left(\frac{\beta}{2}\right)^k \Gamma(1/2-k+(m+n+s)/2)\nonumber \\
& \hspace{3cm}\times\, _2F_1\left(k-\frac{m+n}{2},\frac{1}{2}+k-\frac{m+n}{2}; \frac{1}{2}+k-\frac{m+n+s}{2};\beta \right)
\end{align}
if $m+n+s$ even and $0$ if $m+n+s$ odd. In the case $m=n$ this expression reduces to
\begin{align}
\label{GKH3}
\mathcal{J}^{(H)}_{n,n}(s,\beta) &= \frac{2^{2n}}{\beta^{n+(s+1)/2}}\, \sum_{k=0}^{n} { n \choose
   k}^2
    k! \left(\frac{\beta}{2}\right)^k
    \Gamma(1/2-k+n+s/2)\nonumber \\
    & \hspace{2cm} \times \,  _2F_1\left(k-n,\frac{1}{2}+k-n;
    \frac{1}{2}+k-n-\frac{s}{2};\beta \right)
\end{align}
if $m+n+s$ even and $0$ if $m+n+s$ odd.
\end{theorem}

\begin{proof}
The explicit expression of Hermite polynomials is (see \cite{olver2010})
\begin{equation} \label{her-exp}
H_n(x)=\sum_{k=0}^{[n/2]} \frac{(-1)^k n!}{k! (n-2k)!} (2x)^{n-2k},
\end{equation}
where $[a]$ denotes the greatest integer in a. We also need the well--known linearization formula  for the product of the Hermite orthogonal polynomials (see, for example, \cite{ruiz1997,ruiz1999})
\begin{equation} \label{lin-pro-her}
H_m(x)H_n(x)=\sum_{k=0}^{\min(m,n)} { m \choose k} { n \choose k} 2^k k! H_{m+n-2k}(x).\end{equation}
Then, applying (\ref{her-exp}) and (\ref{lin-pro-her}), we get
\begin{align*}\mathcal{J}^{(H)}_{m,n}(s,\beta)&= \sum_{k=0}^{\min(m,n)} { m \choose k} { n \choose k} 2^k k! \int_{-\infty}^{\infty} x^{s}[\omega^{(H)}(x)]^{\beta}H_{m+n-2k}(x)\, dx. \\
&\hspace{-2cm} = \sum_{k=0}^{\min(m,n)} { m \choose k} { n \choose k} 2^k k!\sum_{j=0}^{[(m+n-2k)/2]} \frac{(-1)^j 2^{m+n-2(k+j)}(m+n-2k)!}{j! (m+n-2k-2j)!}\nonumber\\
& \hspace{-2cm} \times  \int_{-\infty}^{\infty} x^{m+n+s-2(k+j)}e^{-\beta x^2}dx.
\end{align*}
Obviously, the above integral is 0 when $m+n+s$ is odd. Taking into account the following integral
$$
\int_{-\infty}^{\infty} x^{2n}e^{-\beta x^2}dx= \frac{\Gamma(n+1/2)}{ \beta^{n+1/2}}, $$
where $n$ is a nonnegative integer and after some computations we deduce the result for the case $m+n+s$ even.
\end{proof}

\textbf{Applications.} Let us now apply this theorem to the power, Krein-like and logarithmic moments of the Rakhmanov probability density $\rho_{n}^{(H)}(x)$ of Hermite polynomials $H_{n}(x)$ previously defined.
\begin{enumerate}
	\item Power moments.
From \eqref{GKH2}, with $m=n$ and $\beta =1$, one obtains the analytical expression of the corresponding power moments of the Hermite polynomials defined in \eqref{HPM1} as
\begin{align}
\label{PMH1}
\langle x^{s}\rangle_{n}^{(H)} &=  \mathcal{J}^{(H)}_{n,n}(s,1) \nonumber \\
& = 2^{2n}\, \sum_{k=0}^{n} { n \choose
   k}^2
    k! 2^{-k}
    \Gamma(1/2-k+n+s/2)\nonumber \\
    & \times \,  _2F_1\left(k-n,\frac{1}{2}+k-n;
    \frac{1}{2}+k-n-\frac{s}{2};1 \right)
\end{align}
if $m+n+s$ even and $0$ if $m+n+s$ odd.

\item Krein-like moments. The Krein-like moments of Hermite polynomials, obtained from \eqref{GKH2} by making $m=n$, $s=0$ and $\beta =k+1$, are given by
\begin{align}
\label{KMH1}
\langle [\omega(x)]^{k}\rangle_{n}^{(H)}  & =  \mathcal{J}^{(H)}_{n,n}(0,k+1) \nonumber \\
&=  2^{2n}(k+1)^{-n-1/2}\sum_{i=0}^{n} { n \choose
	i}^2
i! \left(\frac{k+1}{2}\right)^i
\Gamma(1/2-i+n)\nonumber \\
& \hspace{2cm} \times \,  _2F_1\left(i-n,\frac{1}{2}+i-n;
\frac{1}{2}+i-n;k+1 \right).
\end{align}

\item Logarithmic moments. The logarithmic moments of the Hermite polynomials can be obtained from the power moments as
\begin{align}
\label{LMH1}
\langle (\log x)^{k}\rangle_{n}^{(H)} &= \frac{d^{k}}{ds^{k}} \langle x^{s}\rangle_{n}^{(H)} \Bigg|_{s=0}.
\end{align}
\end{enumerate}

\subsection{Jacobi polynomials}


\begin{theorem}
Let $P_{m}^{(\alpha,\gamma)}(x)$ denote the Jacobi polynomials orthogonal with respect to the weight function $\omega_{\alpha,\gamma}(x)=(1-x)^{\alpha}(1+x)^{\gamma}$ on $(-1,1)$ \cite{nikiforov1988,olver2010}. Then, the (standard) Krein-like functionals of the Jacobi polynomials defined by 
\begin{equation}
\label{GKJ1}
\mathcal{J}^{(J)}_{m,n}(s,\alpha,\gamma,\beta):=\int_{-1}^{1} x^{s}[\omega_{\alpha,\gamma}^{(J)}(x)]^{\beta}P_{m}^{(\alpha,\gamma)}(x)P_{n}^{(\alpha,\gamma)}(x)\, dx, \quad \alpha,\gamma>-\frac{1}{\beta},\quad s\in\mathbb{Z}_{+}
\end{equation}
can be expressed as
\begin{align}
\label{GKJ2}
\mathcal{J}_{m,n}^{(J)}(s,\alpha,\gamma,\beta)&= \sum_{k=|n-m|}^{n+m}
b_{nmk}^{(\alpha,\gamma)}\sum_{j=0}^{k} c_{kj}^{(\alpha,\gamma)}(\beta)d_{j}^{(\alpha,\gamma)}(s,\beta)\nonumber \\
& \times \frac{2^{\beta(\alpha  +  \gamma) +1} \Gamma (j+\alpha  \beta +1) \Gamma (j+\beta  \gamma +1)}{j! (\beta(\alpha  +  \gamma)+2 j+1) \Gamma (\beta(\alpha  +  \gamma)+j +1)},
\end{align}
with the coefficients
 \begin{align}
 \label{GKJC1}
 b_{nmk}^{(\alpha,\gamma)} &=\frac{k!\, \Gamma(\alpha+\gamma+k+1)
\Gamma(\alpha+\gamma+2n+1)\Gamma(\alpha+\gamma+2n+1)}{n!\,
m!\,\Gamma(\alpha+\gamma+2k+1) \Gamma(\alpha+\gamma+n+1)
\Gamma(\alpha+\gamma+m+1)}\nonumber\\
& \times \sum_{t=\max\{0,k-m\}}^{\min\{k,n\}}
{n \choose t} {m \choose k-t}
\frac{(\alpha+t+1)_{n-t}}{(\alpha+\gamma+n+t+1)_{n-t}} \nonumber  \\
&\times \sum_{w=0}^{m+t-k} {m+t-k \choose w}
\frac{(-1)^w(\alpha+\gamma+2m+1)_{k+w-m-t}(\alpha+k+1)_w}{(\alpha+m+1)_{k+w-m-t}(\alpha+\gamma+2k+2)_w}
\nonumber \\
&\times\, {}_3F_2(t-n,\alpha+\gamma+n+t+1,\alpha+k+w+1;
\alpha+t+1, \alpha+\gamma+2k+w+2;1), \nonumber \\
 c_{kj}^{(\alpha,\gamma)}(\beta) &=\frac{(\alpha+\gamma+k+1)_j
(\alpha+j+1)_{k-j}
(\beta(\alpha+\gamma)+2j+1)\Gamma(\beta(\alpha+\gamma)+j+1)}{(k-j)!\,
\Gamma(\beta(\alpha+\gamma)+2j+2)} \nonumber  \\
&\times {}_3F_2(-k+j,\alpha+\gamma+k+j+1,\alpha \beta+j+1;
\alpha+j+1, \beta(\alpha+\gamma)+2j+2; 1)\quad \text{and} \nonumber \\
d_{j}^{(\alpha,\gamma)}(s,\beta) &=(-1)^{s-j}{s \choose j}{}_2F_1(j-s,\beta \gamma+j+1;
\beta(\alpha+\gamma)+2j+2;2)\, \frac{2^j \,j!\,
\Gamma(\beta(\alpha+\gamma)+j+1)}{\Gamma(\beta(\alpha+\gamma)+2j+1)}.
\end{align}
\end{theorem}

\begin{proof}
First, we observe that
$$
\mathcal{J}^{(J)}_{m,n}(s,\alpha,\gamma,\beta)=\int_{-1}^{1} x^{s}(1-x)^{\alpha \beta} (1+x)^{\gamma \beta} P_{m}^{(\alpha,\gamma)}(x)P_{n}^{(\alpha,\gamma)}(x)\, dx.
$$
Therefore, it is necessary to write down the product $P_{m}^{(\alpha,\gamma)}\,P_{n}^{(\alpha,\gamma)}$ in terms of the polynomials $P_i^{(\alpha\beta,\gamma \beta)}.$ The linearization formula for Jacobi polynomials is given by \cite{savin2017,koepf2010}
$$
P_{m}^{(\alpha,\gamma)}(x)\,P_{n}^{(\alpha,\gamma)}(x)=\sum_{k=|n-m|}^{n+m} b_{nmk}P_{k}^{(\alpha,\gamma)}(x),
$$
where the coefficients $b_{nmk}$ are given in (\ref{GKJC1}). Applying \cite[Lemma 7.1.1]{andrews1999} we obtain
$$
P_{k}^{(\alpha,\gamma)}(x)=\sum_{j=0}^{k}c_{kj}^{(\alpha,\gamma)}(\beta) P_k^{(\alpha \beta,\gamma \beta)}(x),
$$
where  $c_{kj}^{(\alpha,\gamma)}(\beta)$ are given in (\ref{GKJC1}). Therefore, combining the above expressions we get
\begin{equation}
\label{linjacob}
P_{m}^{(\alpha,\gamma)}(x)\,P_{n}^{(\alpha,\gamma)}(x)=\sum_{k=|n-m|}^{n+m} b_{nmk} \sum_{j=0}^{k}c_{kj}^{(\alpha,\gamma)}(\beta) P_j^{(\alpha \beta,\gamma \beta)}(x),
\end{equation}
so,
$$
\mathcal{J}^{(J)}_{m,n}(s,\alpha,\gamma,\beta)=\sum_{k=|n-m|}^{n+m} b_{nmk} \sum_{j=0}^{k}c_{kj}^{(\alpha,\gamma)}(\beta)\int_{-1}^{1} x^{s}(1-x)^{\alpha \beta} (1+x)^{\gamma \beta} P_j^{(\alpha \beta,\gamma \beta)}(x)\, dx.
$$
Finally, using the inversion formula for Jacobi polynomials \cite{ruiz1997,ruiz1999}
\begin{equation*}
x^s= \sum_{\ell=0}^s d_{\ell}^{(\alpha,\gamma)}(s,\beta) P_{\ell}^{(\alpha \beta,\gamma \beta)}(x),
\end{equation*}
we get
\begin{align}
\label{JKF3}
\mathcal{J}^{(J)}_{m,n}(s,\alpha,\gamma,\beta)&=\sum_{k=|n-m|}^{n+m} b_{nmk} \sum_{j=0}^{k}c_{kj}^{(\alpha,\gamma)}(\beta) \sum_{\ell=0}^s d_{\ell}^{(\alpha,\gamma)}(s,\beta) \nonumber \\
&\times \int_{-1}^{1} (1-x)^{\alpha \beta} (1+x)^{\gamma \beta} P_j^{(\alpha \beta,\gamma \beta)}(x) P_{\ell}^{(\alpha \beta,\gamma \beta)}(x)\, dx.
\end{align}
Applying the orthogonality of Jacobi polynomials in Eq. (\ref{JKF3}), we arrive at Eq. (\ref{GKJ2}).

\end{proof}

\textbf{Applications}. Let us now apply this theorem to the power, Krein-like, exponential and logarithmic moments of the Rakhmanov probability density $\rho_{n}^{(J)}(x)$ of the Jacobi polynomials $P_{n}^{(\alpha,\gamma)}(x)$ previously defined.

\begin{enumerate}
	\item Power moments.
From Eq. \eqref{GKJ2}, with $m=n$ and $\beta =1$, one obtains the analytical expression of the corresponding power moments of the Jacobi polynomials as
\begin{align}
\label{PMJ1}
\langle x^{s}\rangle_{n}^{(J)} &  =\mathcal{J}_{n,n}^{(J)}(s,\alpha,\gamma,1)\nonumber \\
&= 2^{(\alpha+\gamma)+1}\sum_{k=0}^{2n}
b_{nk}^{(\alpha,\gamma)}\sum_{j=0}^{k} c_{kj}^{(\alpha,\gamma)}(1)d_{j}^{(\alpha,\gamma)}(s,1)\nonumber \\
& \times \frac{\Gamma(\alpha+j+1)
\Gamma(\gamma+j+1)}{((\alpha+\gamma)+2j+1)\,  j! \,
\Gamma((\alpha+\gamma)+j+1)}
\end{align}
where one can obtain the values of the coefficients, $b_{nk}^{(\alpha,\gamma)}$, $c_{kj}^{(\alpha,\gamma)}(1)$ and $d_{j}^{(\alpha,\gamma)}(s,1)$ by just substituting $m=n$ and $\beta=1$ in their corresponding definitions. We omit them since they hardly get simplified.

\item Krein-like moments.
The Krein-like moments of Jacobi polynomials, obtained from \eqref{GKJ2} by making $m=n$, $s=0$ and $\beta =k+1$, are given by
\begin{align}
\label{KMJ1}
\langle [\omega(x)]^{k}\rangle_{n}^{(J)}  & =  \mathcal{J}^{(J)}_{n,n}(0,\alpha,\gamma,k+1) \nonumber \\
& = 2^{(k+1)(\alpha+\gamma)+1}\sum_{i=0}^{2n}
b_{ni}^{(\alpha,\gamma)}\sum_{j=0}^{i} c_{ij}^{(\alpha,\gamma)}(k+1)d_{j}^{(\alpha,\gamma)}(0,k+1)\nonumber\\
& \times \frac{\Gamma(\alpha (k+1)+j+1)
\Gamma(\gamma (k+1)+j+1)}{((k+1)(\alpha+\gamma)+2j+1)\, j!\,
\Gamma((k+1)(\alpha+\gamma)+j+1) }
\end{align}
where one can obtain the values of the coefficients, $b_{ni}^{(\alpha,\gamma)}$, $c_{ij}^{(\alpha,\gamma)}(k+1)$ and $d_{j}^{(\alpha,\gamma)}(0,k+1)$ by just substituting $m=n$, $s=0$ and $\beta=k+1$ in their corresponding definitions. Again, we omit them since they hardly get simplified.

\item Logarithmic moments.

The logarithmic moments of the Jacobi polynomials are obtained from the power moments as
\begin{align}
\label{LMJ1}
\langle (\log x)^{k}\rangle_{n}^{(J)} &= \frac{d^{k}}{ds^{k}} \langle x^{s}\rangle_{n}^{(J)} \Bigg|_{s=0}.
\end{align}

\end{enumerate}

\section{Conclusions}

In this work three different approaches are proposed to compute the integral functionals with kernel of the type $[\omega(x)]^{\beta}\,x^{s}\, p_{m_{1}}(x)\,\ldots \,p_{m_{r}}(x)$, where $[p_{m}]$ denote hypergeometric polynomials orthogonal (HOPs) with respect to the weight function $\omega(x)$ on a real interval. They are called by generalized Krein-like functionals because of their close connection to the moment problems of power and Markov types early studied by M. G. Krein. These functionals are frequently encountered in numerous analytical and computational problems of theoretical physics and applied mathematics ranging from approximation theory, information theory of quantum systems up to quantum physics of atomic, molecular and nuclear physics and intermediate and high energy physics.

In the first approach the general case of a finite number of HOPs has been considered, showing that the associated Krein-like functionals can be expressed in a compact way by use of some multivariate hypergeometric functions of Lauricella and Srivastava-Daoust types. In the second and third approaches we consider only two HOPs for simplicity, obtaining the associated Krein functionals explicitly in terms of its coefficients of the hypergeometric equation and the polynomial degrees and/or the parameters of the corresponding weight function, respectively. 

Finally, the usefulness of the resulting expressions is illustrated by their application to compute a number of mathematical quantities like the moments of power, Krein, exponential and logarithmic types of the Rakhmanov probability density of the involved HOPS, which are not only relevant \textit{per se} but also because they describe numerous fundamental and experimentally accessible physical quantites of quantum systems. 

\section*{Acknowledgments}

JJMB is partially supported by Ministerio de Ciencia, Innovaci\'on y Universidades of Spain and European Regional Development Fund (MCIU-FEDER) through the grant  MTM2014-53963-P,  Junta de Andaluc\'{i}a through the Research Group FQM-0229 (belonging to Campus of International Excellence CEIMAR) and the research center CDTIME. IVT and JSD are partially supported by MCIU-FEDER through grant FIS2014-54497P, FIS2014-59311P, FIS2017-89349P and by Junta de Andaluc\'{i}a through the Research Group FQM-207.

\end{document}